\documentclass[12pt]{article}
\usepackage[ruled]{algorithm2e}
\usepackage{algorithmic, setspace}
\usepackage{amsmath}
\usepackage{amssymb}  
\usepackage{amsthm}
\usepackage{bm}
\usepackage{graphicx}
\usepackage{enumerate}
\usepackage{natbib}
\usepackage{url} 
\usepackage{xcolor}
\usepackage{tocloft}   
\usepackage{bibunits}  
\usepackage{hyperref}

\newcommand{\blind}{1}

\newcommand{\N}{\mathbb{N}}
\newcommand{\R}{\mathbb{R}}

\DeclareMathOperator{\E}{\mathbb{E}}        
\DeclareMathOperator{\var}{var}             
\DeclareMathOperator{\pr}{Pr}                    
\newcommand{\ind}[1]{\text{I}\left({#1}\right)}  

\newcommand{\indsim}{\stackrel{ind}{\sim}}       
\newcommand{\iidsim}{\stackrel{iid}{\sim}}       

\def\Beta{\hbox{Beta}}
\def\Ber{\hbox{Ber}}
\def\CRP{\hbox{CRP}}

\def\Ga{\hbox{Ga}}
\def\InvGa{\hbox{Inv-Ga}}
\def\Normal{\hbox{N}}

\newtheorem{definition}{Definition}
\newtheorem{proposition}{Proposition}
\newtheorem{proposition*}{Proposition}

\newcommand{\ARI}{\text{ARI}}
\newcommand{\tRPM}{\text{tRPM}}
\newcommand{\smRPM}{\text{smRPM}}

\begin{document}

\addtocontents{toc}{\protect\setcounter{tocdepth}{-1}}

\def\spacingset#1{\renewcommand{\baselinestretch}%
{#1}\small\normalsize} \spacingset{1}


\if1\blind
{
  \title{\bf  Bayesian local clustering of functional data via semi-Markovian random partitions}
  \author{Giovanni Toto\hspace{.2cm}\\
    Department of Statistics and Data Science, University of Texas at Austin, \\Austin, TX, USA\\
    and \\
    Antonio Canale \\
    Department of Statistics, University of Padova, Padua, Italy}
  \maketitle
} \fi

\if0\blind
{
  \bigskip
  \bigskip
  \bigskip
  \begin{center}
    {\LARGE\bf Bayesian local clustering of functional data via semi-Markovian random partitions}
\end{center}
  \medskip
} \fi

\bigskip
\begin{abstract}
We introduce a Bayesian framework for indirect local clustering of functional data, leveraging B-spline basis expansions and a novel dependent random partition model. By exploiting the local support properties of B-splines, our approach allows partially coincident functional behaviors, achieved when shared basis coefficients span sufficiently contiguous regions. This is accomplished through a cutting-edge dependent random partition model that enforces semi-Markovian dependence across a sequence of partitions. By matching the degree of the B-spline basis with the semi-Markovian dependence structure, the proposed model serves as a highly flexible prior, enabling efficient modeling of localized features in functional data. Furthermore, we extend the utility of the dependent random partition model beyond functional data, demonstrating its applicability to a broad class of problems where sequences of dependent partitions are central, and standard Markovian assumptions prove overly restrictive. Empirical illustrations, including analyses of simulated data and tide level measurements from the Venice Lagoon, showcase the effectiveness and versatility of the proposed methodology.

\end{abstract}

\noindent%
\vfill

\newpage
\spacingset{1.9} 

\section{Introduction} \label{sec:intro}
Functional data analysis \citep{ramsay_FunctionalDataAnalysis_2005, morris_FunctionalRegression_2015} 
addresses classical statistical challenges, including regression, classification, and clustering, in contexts where the data are naturally represented as smooth functions---such as curves or surfaces---defined over a continuous domain. For example, functional data clustering \citep[see][for a review]{zhang_ReviewClusteringMethods_2023} involves identifying homogeneous groups within a sample of functional data. 
Clustering in the functional context is inherently more complex than its Euclidean counterpart, due to the unique characteristics of function spaces.
In fact, functional observations often display diverse behaviors across different regions of their domains, making it crucial to clearly define the objectives of the clustering process, whether the aim is to capture the global trends or localized variations within the data.
To illustrate this, consider the toy example reported in Figure~\ref{fig:clustering}. While three distinct groups can be clearly identified when observing the global behavior of the functional data, different behaviors may emerge when focusing on subsets of the domain.
For instance, all functional observations exhibit the same behavior in the bounds of their domain, while multiple behaviors are observed in the center.
In line with this, we define global clustering as the approach where each functional observation is entirely assigned to a single cluster. In contrast, local clustering explores how the partitioning of functional data evolves along their domains, allowing a single functional observation to belong to different clusters depending on the specific point in its domain being evaluated.

\begin{figure}
    \centering
    \includegraphics[width=0.9\linewidth,height=5cm]{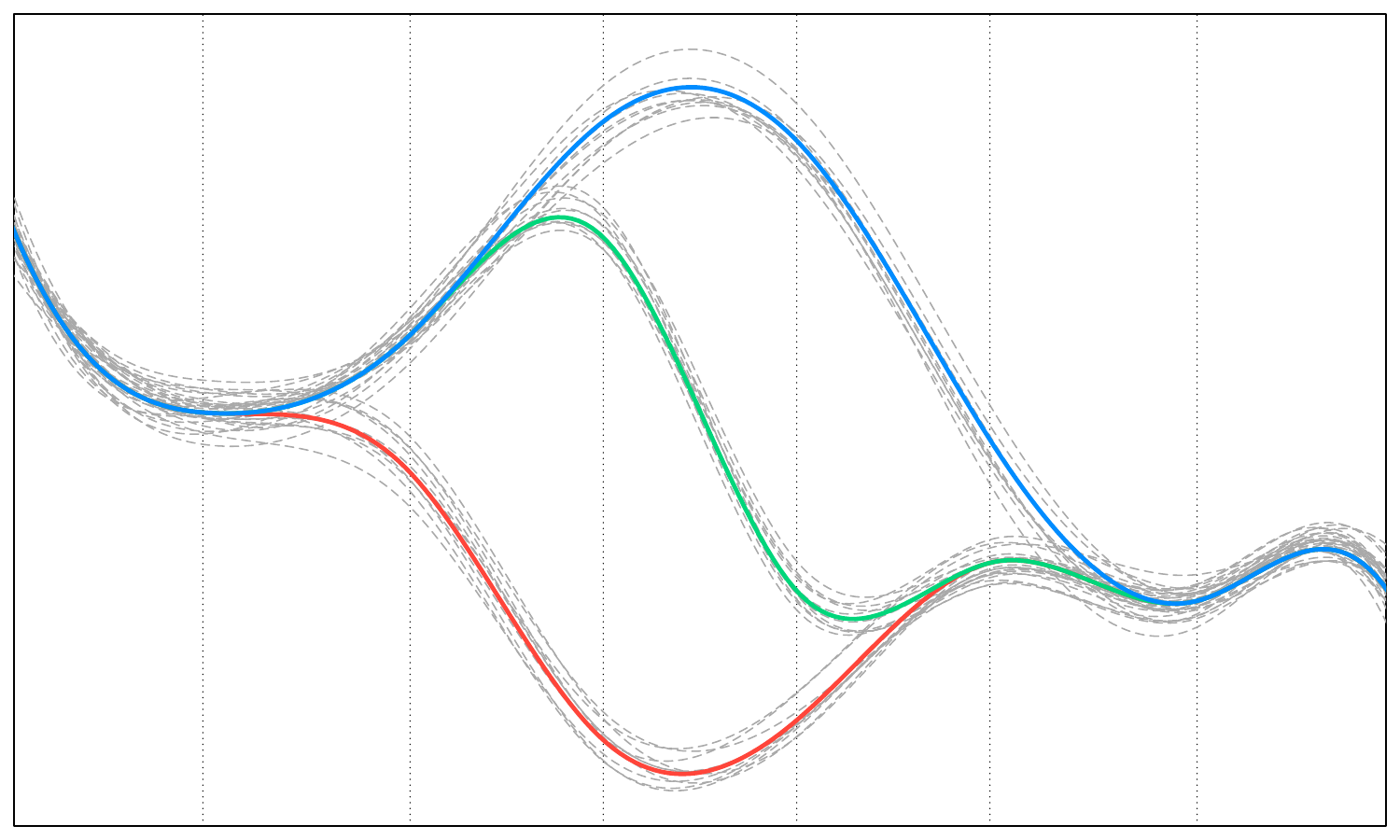}
    \caption{Illustrative example of local-global functional clustering. Grey dashed lines represent functional data. Colored thick lines represent three global clusters that collapse in certain subregions of the domain. }
    \label{fig:clustering}
\end{figure}

Early contributions to this topic, under a Bayesian context, have relied on variations of the Dirichlet Process \citep[DP,][]{ferguson_BayesianAnalysisNonparametric_1973}, such as the hybrid Dirichlet Process \citep{petrone_HybridDirichletMixture_2009} and the nested Hierarchical Dirichlet Process \citep{nguyen_InferenceGlobalClusters_2010}.
\citet{petrone_HybridDirichletMixture_2009} proposed a Bayesian nonparametric mixture model which accounts for global and local heterogeneity by representing individual curves as recombinations of a set of canonical curves.
\citet{nguyen_DirichletLabelingProcess_2011} undertook a detailed analysis of the Dirichlet labeling process introduced in \cite{petrone_HybridDirichletMixture_2009} and proposed a more computationally tractable alternative labeling process.
Later, the same authors \citep{nguyen_BayesianNonparametricModeling_2014} introduced a novel approach to solve the functional ANOVA problem that extends and merges together the ideas in \citet{teh_HierarchicalDirichletProcesses_2006} and \citet{nguyen_InferenceGlobalClusters_2010}.
Recently, \citet{hellmayr_PartitionDirichletProcess_2021} introduced the concepts of marginal and joint labeling processes, used them to adapt the models in \citet{gelfand_BayesianNonparametricSpatial_2005} and \citet{duan_GeneralizedSpatialDirichlet_2007} to functional data, hence defining the Functional Dirichlet process (FDp) and the Generalized Functional Dirichlet process (GFDp), and proposed a middle-ground specification denoted as the Partitioned Functional Dirichlet process (PFDp).

Another stream of  research exploits functional data approximations via basis expansion.
\citet{suarez_BayesianClusteringFunctional_2016} assumed that functional observations are noisy realizations of a signal function that can be approximated by a wavelet basis expansion. 
An \textsl{ad-hoc} prior on the wavelet parameters allows for similar but not exactly identical signal functions, however local clustering can be performed applying any clustering approach to a similarity matrix computed from posterior samples of the wavelet parameters.
To avoid a two-step clustering approach, \citet{paulon_BayesianSemiparametricHidden_2024, fan_BayesianSemiparametricLocal_2024} exploited the local property of B-spline basis expansions by defining basis-specific partitions on the expansion parameters.
Both works induce dependence across partitions letting the cluster allocations evolve according to a Hidden Markov Model (HMM). 

As an alternative to the HMM approach, a latent (dependent) random partition modeling framework can be employed, consistent with widely adopted Bayesian model-based clustering proposals. 
Following this approach, the local clustering of functional data is associated with a sequence of random partitions that evolves through the functions' domain.
A series of works build on the idea that a specific Random Partition Model (RPM) is specified for the initial partition, and that auxiliary variables are introduced to guide its evolution over time.
For example, \citet{page_DependentModelingTemporal_2022} proposed the temporal Random Partition Model (tRPM), in which adjacent partitions are linked through unit- and time-specific binary latent auxiliary variables specifying whether units can or cannot be reallocated when moving from a time point to the next one.
\citet{romano_BayesianLocalClustering_2025,romano_DependentStochasticBlock_2025} are two recent pre-prints showing that tRPM and stRPM, its spatial-aware extension, can be successfully used as building blocks in complex hierarchical models.
\citet{liang_BayesianNonparametricApproach_2024} focused on information sharing across different times, and proposed a hierarchical temporal dependent functional clustering approach. The partitions are induced by a Hierarchical Dirichlet Process \citep{teh_HierarchicalDirichletProcesses_2006} whose hierarchy allows for shared atoms, i.e., cluster-specific parameters, at different times.
These formulations with unit-specific auxiliary variables provide a high degree of flexibility, in the sense that adjacent partitions may be very similar or very dissimilar depending on the values assumed by the latent variables, however this may also hinder the ability of the model to detect changepoints.
To this end, \citet{giampino_LocalLevelDynamic_2024} proposed an alternative formulation in which a single auxiliary variable is introduced at each time point that specifies whether the entire partition can change or not.
Compared to HMMs, these approaches offer significantly greater flexibility such as flexible inclusion of prior knowledge, integration of covariate dependencies, ability to specify an alternative RPM for the initial partition, and adaptability through modifications of the auxiliary variables' distribution.

In the broader context of dependent random partitions, auxiliary variables are used in the informed partition model proposed by \cite{paganin_InformedPartitionModels_2025} to link the partitions explored by the model to an ``initial'' one encoding prior knowledge. 
Among the literature focusing on linking random partitions to a reference or baseline partition, we recall the pioneering Location–Scale Partition \citep[LSP;][]{smith_DemandModelsRandom_2020} and the Centered Partition Process \citep[CPP;][]{paganin_CenteredPartitionProcesses_2021}, as well as the more recent Shrinkage Partition (SP) distribution \citep{dahl_DependentRandomPartitions_2025}.

In this article, we present a novel Bayesian hierarchical framework for performing local clustering of functional data, leveraging a novel process for random partition modeling. Our approach is motivated by the local properties of B-spline basis representations, which serve as the foundation of our methodology. The latter is coupled with a generalization of the tRPM, which we term semi-Markovian Random Partition Model (smRPM), in which the dependence structure matches the degree of the B-spline basis.   

Notably, although  motivated by the local functional clustering problem, the smRPM represents a substantial advancement with far-reaching implications for the Bayesian Nonparametrics (BNP) literature on dependent random partitions.
By relaxing the  assumptions of Markovian dependence between partitions, the smRPM offers a flexible modeling alternative with broad potential impact extending beyond the motivating functional data context.

The remainder of the article is organized as follows:
Section~\ref{sec:modelfunctionaldata} introduces the idea of indirect local clustering of functional data through basis-specific partitions in a basis expansion, paving the way to a novel class of random partition models with higher-order dependence.
This is presented in Section~\ref{sec:smRPM} which also presents results enabling posterior inference via Gibbs sampling.
Simulation studies assessing the behavior of the proposed smRPM are presented in Section~\ref{sec:simulations}.
Section~\ref{sec:venice} provides an illustrative application of local clustering for functional data, focusing on the analysis of tide level measurements in the Venice Lagoon over time.  Section \ref{sec:conclusion} offers concluding remarks.

\section{Local clustering of functions via basis expansion}  
\label{sec:modelfunctionaldata}

Let $y_i=(y_i(x), x\in \mathcal{D})$, $i=1,\ldots,n$, be random curves defined on the domain $\mathcal{D}\subseteq\R$.
To represent these curves, we assume
\begin{equation} \label{eq:basis_expansion}
    y_i(x) = \mathbf{b}(x)^\top\bm{\theta}_i = \sum_{k=1}^K b_k(x) \theta_{ik},
\end{equation}
where $\bm{\theta}_i=(\theta_{i1},\ldots,\theta_{iK})^\top\in\R^{K}$ is a curve-specific basis coefficient vector, and $\mathbf{b}(\cdot)$ is a $d$-degree B-spline basis with $K$ basis functions, associated with $K-d+1$ equispaced distinct knots. It is well known that the number of knots (or, equivalently, the number of basis functions) plays a crucial role in determining the flexibility and smoothness of the resulting functions. In our motivating application to functional data clustering, this choice is particularly relevant, as it directly influences the scale at which local similarities between curves can be identified, as discussed later.
 
A common strategy for clustering functions represented through basis expansions involves clustering the associated basis coefficients. Under a Bayesian mixture model formulation, one specifies
\begin{equation} \label{eq:mixture1}
    \bm{\theta}_i \sim P, \quad P = \sum_{j\geq1 } \pi_j \delta_{\bm{\theta}_j^*}
\end{equation} 
where $P$ is a finite or infinite-dimensional random mixing measure with random weights $\{\pi_j\}_{j\geq1}$ and random atoms $\{\bm{\theta}_j^*\}_{j\geq1}$; $\delta_{\bm{\theta}_j^*}$ denotes a Dirac measure on $\bm{\theta}_j^*$. Typical specifications assume independently and identically distributed vectors of basis coefficients from a common base measure $P_0$. For example, under a B-spline specification, a suitable choice for the base measure $P_0$ is the penalized prior introduced by \citet{lang_BayesianPSplines_2004}, which penalizes the differences between consecutive basis coefficients to ensure the resulting curves are smooth. 
Consistent with this, two curves $y_i$ and $y_{i'}$ are considered to belong to the same cluster if every element of $\bm{\theta}_i$ matches its counterpart in $\bm{\theta}_{i'}$.
This scenario has a non-zero prior probability, as the mixture representation in \eqref{eq:mixture1} admits ties in the $\bm{\theta}_i$ due to the almost sure discrete nature of $P$.
Consequently, one typically expresses the relationship as $\bm{\theta}_{i} = \bm{\theta}_{c_i}^*$, where $c_i \in \{1,2, \dots\}$ is a global label identifying the cluster to which unit $i$ belongs.

In our motivating local clustering problem, however, we generalize the global $c_i$ with a vector $\mathbf{c}_i= (c_{i1},\ldots,c_{iK})^\top$  of local cluster labels. Consistently, let   $\mathbf{C}$ be the $n\times K$ matrix containing cluster membership labels for each basis, defined as
\begin{equation} \label{eq:cluster_matrix}
    \mathbf{C} =
    \begin{bmatrix}
    \mathbf{c}_{1} \\
    \vdots \\
    \mathbf{c}_{n} \\
    \end{bmatrix} =
    \begin{bmatrix}
    c_{11} & \ldots & c_{1K} \\
    \vdots & c_{ik} & \vdots \\
    c_{n1} & \ldots & c_{nK} \\
    \end{bmatrix}.
\end{equation}
Given $\mathbf{c}_{i}$, $\bm{\theta}_i =(\theta^*_{1,c_{i1}},\ldots,\theta^*_{K,c_{iK}})^\top$, where $\theta^*_{kj}$ is the B-spline basis coefficient of cluster $j$ at basis $k$. 

Three aspects are crucial to satisfying the desired local clustering property for functional data expressed through \eqref{eq:basis_expansion}: i) for each $i$, there should be a dependence between contiguous $c_{ik}$, ii) this dependence should result in actual functional clustering rather than merely clustering the basis coefficients, and iii) the resulting curves should exhibit sufficient regularity in terms of smoothness. Issues i) and ii) are addressed through tailored random partition models defining the latent variables in $\mathbf{C}$, while issue iii) is resolved by appropriately defining the prior process generating the atoms $\theta_{kj}^*$.

While any dependent random partition model, including the tRPM of \citet{page_DependentModelingTemporal_2022}, may address issue i), issue ii) requires specific justification based on the characteristics of the basis expansion employed. Notably, an important property of $d$-degree B-splines is that a function is fully determined within the $k$th subinterval $\mathcal{D}_k \subset \mathcal{D}$, where the subintervals define the partition of $\mathcal{D}$ induced by the B-spline knots, by its $d+1$ contiguous parameters \( \theta_{i,k}, \ldots, \theta_{i,k+d} \). Therefore, for two curves to be \textsl{locally} identical, it suffices that they share $d+1$ contiguous parameters. Consequently, it becomes clear that the tRPM of \citet{page_DependentModelingTemporal_2022}, which ensures continuity in cluster membership across transitions between adjacent partitions, is insufficient to match the degree of the B-spline basis. Motivated by this limitation, we propose a more flexible semi-Markovian Random Partition Model (smRPM) in Section~\ref{sec:smRPM}. It is important to note that, as previously discussed, the choice of $K$ determines the size of the subintervals $\mathcal{D}_k$, the latter being equal to the total length of $\cal D$ divided by $(K-d)$, which in turn sets the minimum scale at which two or more curves can overlap. In the following, we assume that $K$ is fixed. Such a choice should be made on a case-by-case basis, guided by pragmatic considerations specific to each application.

Before constructing the proposed smRPM, we briefly discuss how we  propose to address point iii). The prior distribution for the $\theta^*_{kj}$ is defined conditionally on the local clusters $\mathbf{C}$. As in \citet{lang_BayesianPSplines_2004}, we enforce contiguous $\theta^*_{kj}$ to be similar to prevent erratic curve behavior. However, under our local clustering framework, the notion of contiguity extends beyond simple adjacency because different clusters at basis $k-1$ may merge into a single cluster at basis $k$. Accordingly, the general $\theta^*_{kj}$ should remain close to the previous $\theta^*_{k-1,l}$ for each cluster $l$ that is absorbed into cluster $j$ at basis $k$.

A relatively simple approach  to manage this scenario is to assume that the prior expectation of the B-spline parameter for cluster  $j$ at basis $k$, i.e., $\E[\theta^*_{kj}]$, is defined as the average of the cluster parameters at basis $k-1$ associated with clusters that are reallocated to cluster $j$ at least once when transitioning from basis $k-1$ to $k$. A specification grounded in this intuition is given by:
\begin{equation} \label{eq:model_theta}
\begin{split}
\theta^*_{1j} \mid \tau^2 &\indsim \Normal(0, \tau^2), \quad j=1,\ldots, J_1,\\
    \theta^*_{kj} \mid \bm{\theta}^*_{k-1},\mathbf{C},\phi,\tau^2 &\indsim \Normal\left(\frac{\phi}{|\mathcal{C}_{k-1}^{(\to j)}|}\sum_{l\in \mathcal{C}_{k-1}^{(\to j)}}\theta^*_{k-1,l},\tau^2\right), \quad j=1,\ldots, J_k, \quad k=2,\ldots,K,
\end{split}
\end{equation}
where $J_k$ is the number of clusters at basis $k$, and $\mathcal{C}_{k-1}^{(\to j)}$ denotes the set $\{l\in\{1,\ldots,J_{k-1}\}: \sum_{i=1}^n\ind{c_{i,k-1}=l, c_{i,k}=j}>0\}$ containing the $|\mathcal{C}_{k-1}^{(\to j)}|$ indices of clusters which are reallocated to cluster $j$ at least once when moving from basis $k-1$ to $k$. 
The hierarchical Bayesian model is completed by specifying the prior distributions for $\phi$ and $\tau^2$, i.e., $\phi \sim \Normal(m_0,s^2_0)$ and $\tau^2 \sim \InvGa(a_\tau,b_\tau)$, where $\InvGa(a,b)$ denotes an inverse Gamma distribution with shape $a>0$ and scale $b>0$. 

Before concluding this section, we define the model for the observed data. In fact,  we typically observe noisy realizations of the random curves at a finite set of coordinates, which may vary across functions, $\bm{x}_i=(x_{i1},\ldots,x_{iT_i})$. Assuming  ${Y}_i(x_{it})$ to be the noisy realization of  curve  $y_i$ at point $x_{it}$ we let
\begin{equation} \label{eq:model_curves}
    Y_i(x_{it}) \mid \bm{\theta}^*,\mathbf{c}_{i},\sigma^2 \indsim \Normal\left( \sum_{k=1}^K b_k(x_{it}) \theta^*_{k,c_{ik}}, \sigma^2 \right).
\end{equation}
Additionally, let $\mathbf{y}_i(\bm{x}_i)=(y_i(x_{i1}), \ldots, y_i(x_{iT_i}))^\top$ be the collection of the points in which curve $i$ is observed,  for $i=1,\ldots,n.$ An additional prior for $\sigma^2$ is assumed as 
$\sigma^2 \sim \InvGa(a_\sigma,b_\sigma)$.

\section{Semi-Markovian Random Partition Models} \label{sec:smRPM}
We propose a general prior process to model the evolution of latent partitions as a domain index $k$ evolves. This can be generally thought of as time or, consistently with our functional data motivation, as the B-spline basis index.

We start by introducing some general notation. Let $\rho_{1},\ldots,\rho_{K}$ represent the sequence of the evolving partitions with $\rho_k=\{S_{k1},\ldots,S_{kJ_k}\}$ denoting the partition at domain index $k$ of the $n$ objects into $J_k$ clusters, $S_{kj}\subseteq\{1,\ldots,n\}$, $k=1,\ldots,K$. Notably, this is related to the notation introduced in the previous section, and specifically considering that  $c_{ik}=j$ implies $i\in S_{kj}$. 
In what follows, we use $\rho_k$ and $(c_{1k},\ldots,c_{nk})$ interchangeably.

Following ideas first introduced in \citet{page_DependentModelingTemporal_2022}, we  introduce auxiliary variables which explicitly model the evolution of the partitions at different $k=1,\ldots,K$. Instead of modeling the transition from $k$ to the subsequent $k+1$, we introduce a notion of persistency through the partitions of order $d_\rho\geq1$. Consistent with this, let $\gamma_{ik}$ be the binary latent variable specifying whether the $i$th unit cannot be considered for possible cluster reallocation for the next $d_\rho$ consecutive index transitions, i.e.,
\begin{equation} \label{eq:gamma_def}
    \gamma_{ik} =
    \begin{cases}
        1 & \text{if unit $i$ \emph{cannot} be reallocated in the $d_\rho$  consecutive indices after $k-1$} \\
        0 & \text{otherwise}
    \end{cases}.
\end{equation}
Under this setting, $\gamma_{ik}=1$ means that the cluster assigned to unit $i$ at index $k-1$ remains the same for the $d_\rho$ consecutive indices. The $n$ auxiliary variables at index $k$, denoted as $\bm{\gamma}_k=(\gamma_{1k},\ldots,\gamma_{nk})$, influence the partitions at the next $d_\rho$ indices, thus inducing $d_\rho$-order persistency.

By construction, at index $k=1$,  we set $\bm{\gamma}_1=\mathbf{0}_n$, and assume that $\pr(\rho_1)$ is an Exchangeable Partition Probability Function (EPPF).
In this paper, we consider the simple  EPPF induced by a DP, often referred to as a Chinese Restaurant Process \citep[CRP,][]{pitman_DevelopmentsBlackwellMacqueenURN_1996},
\begin{equation*}
    \pr(\rho_1 \mid M) = \frac{1}{\prod_{i=0}^{n-1}(M+i)} \prod_{j=1}^{J_1} \left\{ M\cdot(|S_{1j}|-1)! \right\},
\end{equation*}
where $M>0$ is a concentration parameter.
Under $\rho_1\sim\CRP(M)$, the predictive rule for unit $i$ given the partition $\rho_1$ with the $i$th unit removed, $\rho_1^{(-i)}$, is
\begin{equation*}
    \pr(c_{i1} = j \mid \rho^{(-i)}_1, M) \propto
    \begin{cases}
        |S^{(-i)}_{1j}| & \text{if $j=1,\ldots,J_1^{(-i)}$} \\
        M & \text{if $j=J_1^{(-i)}+1$}
    \end{cases},
\end{equation*}
where $S^{(-i)}_{1j}$ is the $j$th cluster of $\rho_1$ with the $i$th unit removed, and $J_1^{(-i)}$ is the number of clusters in $\rho_1^{(-i)}$.
The joint distribution of the partitions and of the auxiliary variables is
\begin{equation} \label{eq:smRPM}
\pr(\bm{\gamma}_1,\rho_1,\ldots,\bm{\gamma}_K,\rho_K) = \pr(\rho_1) \prod_{k=2}^K \pr(\bm{\gamma}_k)\pr(\rho_k \mid \bm{\gamma}_{k-d_\rho+1:k}, \rho_{k-1}),
\end{equation}
where $\bm{\gamma}_{k-d_\rho+1:k}=(\bm{\gamma}_{k-d_\rho+1},\ldots,\bm{\gamma}_k)$ contains the auxiliary variables which influence unit reallocation when moving from index $k-1$ to $k$. Similarly, we  denote the sequence of the elements related to unit $i$ in $\bm{\gamma}_{k-d_\rho+1:k}$ as $\bm{\gamma}_{i,k-d_\rho+1:k}=(\gamma_{i,k-d_\rho+1},\ldots,\gamma_{i,k})^\top$.
We set $\gamma_{ik}=0$ for $k\leq0$: this caveat is to avoid having to provide ad-hoc definitions for random variables with index $k\leq d_\rho$. The idea is that non-existent auxiliary variables do not influence the reallocation of units.

The joint distribution in \eqref{eq:smRPM} is made of two key ingredients, namely the conditional probability of the partition $\pr(\rho_k \mid \bm{\gamma}_{k-d_\rho+1:k}, \rho_{k-1})$ and the probability $\pr(\bm{\gamma}_k)$ which are discussed below.
Before going on, it is worth noting that the distribution of $\rho_k$ could incorrectly suggest that the dependence between partitions is Markovian only. This is not the case since the $d_\rho$-order dependence is induced by the auxiliary variables in $\bm{\gamma}_{k-d_\rho+1:k}$.
Each auxiliary variable $\gamma_{ik}$ influences the reallocation of the $i$th unit for the $d_\rho$ consecutive partitions after $k-1$.
Specifically, $\gamma_{ik}=1$ implies that the $i$th unit belongs to the same cluster in all consecutive partitions $\rho_{k-1}, \ldots, \rho_{k+d_\rho-1}$ by locking the reallocation in each transition between those partitions.
This sequence of $d_\rho$ Markovian constraints can be seen as a unique semi-Markovian constraint.

The distribution of the partition at index $k$ depends on the partition at the previous index, $\rho_{k-1}$, and the auxiliary variables in $\bm{\gamma}_{k-d_\rho+1:k}$.
These two quantities determine the partitions that can be observed at index $k$, or, using the terminology introduced in \citet{page_DependentModelingTemporal_2022}, the set of compatible partitions at index $k$.  Specifically, the $i$th unit can be reallocated if all elements in $\bm{\gamma}_{i,k-d_\rho+1:k}$ are equal to 0 and cannot be reallocated if at least one element is equal to 1 as formalized in the following definition. 
\begin{definition} \label{def:comp_part}
    The partitions $\rho_{k-1}$ and $\rho_k$ are compatible with respect to $\bm{\gamma}_{k-d_\rho+1:k}$ if $\rho_k$ can be obtained from $\rho_{k-1}$ by reallocating units as indicated by $\bm{\gamma}_{k-d_\rho+1:k}$, that is, those units $i$ such that $\max\{\bm{\gamma}_{i,k-d_\rho+1:k}\}=0$ for $i=1,\ldots,n$.
\end{definition}

To check whether $\rho_{k-1}$ and $\rho_k$ are compatible with respect to the vector $\bm{\gamma}_{k-d_\rho+1:k}$, we simply need to verify that the units that cannot be reallocated when moving from index $k-1$ to $k$ have the same cluster assigned in the two partitions.
Formally, let $\mathcal{R}_k = \{i\in\{1,\ldots,n\}: \max\{\bm{\gamma}_{i,k-d_\rho+1:k}\}=1 \}$ be the indices of units that are not reallocated when moving from index $k-1$ to $k$ and denote by $\rho^{\mathcal{R}_k}_k$ the reduced partition at index $k$ containing only the units specified in $\mathcal{R}_k$, then $\rho_{k-1}$ and $\rho_k$ are compatible with respect to $\bm{\gamma}_{k-d_\rho+1:k}$ if and only if $\rho^{\mathcal{R}_k}_{k-1}=\rho^{\mathcal{R}_k}_k$. 

Denoting by $\mathcal{P}$ the set of all partitions and by $\mathcal{P}_{\mathcal{R}_k} = \{\rho_k\in \mathcal{P}:\rho_{k-1}^{\mathcal{R}_k}=\rho_{k}^{\mathcal{R}_k}\}$ the set of partitions that are compatible with $\rho_{k-1}$ based on $\bm{\gamma}_{k-d_\rho+1:k}$, then $\pr(\rho_k\mid \bm{\gamma}_{k-d_\rho+1:k}, \rho_{k-1})$ is a random partition distribution with support in $\mathcal{P}_{\mathcal{R}_k}$ defined as
\begin{equation} \label{eq:rho_prior}
    \pr(\rho_k=\lambda \mid \bm{\gamma}_{k-d_\rho+1:k}, \rho_{k-1}) = \frac{\pr(\rho_k=\lambda)\ind{\lambda\in \mathcal{P}_{\mathcal{R}_k}}}{\sum_{\lambda'\in \mathcal{P}}\pr(\rho_k=\lambda')\ind{\lambda'\in \mathcal{P}_{\mathcal{R}_k}}},
\end{equation}
where $\pr(\rho_k=\lambda)$ is the EPPF of $\rho_1$ evaluated at $\lambda\in \mathcal{P}$ and $\ind{A}$ denotes the indicator function.

We complete the specification of the model by assigning a prior to the auxiliary variables.
As previously done in \citet{page_DependentModelingTemporal_2022}, we assume simple independence, i.e.,
\begin{equation} \label{eq:gamma_prior_tRPM}
    \gamma_{ik}\mid\alpha_k \indsim \Ber(\alpha_k)
\end{equation}
where $\alpha_k$ in turn follows a Beta distribution, $\alpha_k\iidsim\Beta(a_\alpha,b_\alpha)$.
These act as dependence parameters, with high $\alpha_k$ implying that units are more likely not to be reallocated in all transitions between index $k-1$ and $k+d_\rho-1$.

We term the model induced by \eqref{eq:smRPM} equipped with \eqref{eq:rho_prior} and \eqref{eq:gamma_prior_tRPM} the semi-Markovian Random Partition Model of order $d_\rho$, dependence parameter vector $\bm{\alpha}=(\alpha_1,\ldots,\alpha_K)^\top\in(0,1)^K$, and concentration parameter $M>0$, denoted by $\smRPM_{d_\rho}(\bm{\alpha},M)$.
Further assuming that the first partition follows a CRP, the resulting joint probabilistic structure of smRPM is
\begin{equation*}
    \begin{split}
        \rho_1 &\sim \CRP(M), \\
        \pr(\rho_k=\lambda \mid \bm{\gamma}_{k-d_\rho+1:k}, \rho_{k-1}) &= \frac{\pr(\rho_k=\lambda)\ind{\lambda\in \mathcal{P}_{\mathcal{R}_k}}}{\sum_{\lambda'\in \mathcal{P}}\pr(\rho_k=\lambda')\ind{\lambda'\in \mathcal{P}_{\mathcal{R}_k}}}, \\
        \gamma_{ik} \mid \alpha_k &\indsim \Ber(\alpha_k).
    \end{split}
\end{equation*}
Notably, if $d_\rho=1$, then $\smRPM_{d_\rho}(\bm{\alpha},M)$ simplifies to $\tRPM(\bm{\alpha},M)$.

\subsection{Conditional posterior distributions}
In this section, we discuss how to perform posterior inference via Gibbs sampling, by deriving the full conditional distributions for the partitions and the auxiliary variables borrowing ideas from \citet{neal_MarkovChainSampling_2000}.
We start by stating the following two propositions, which extend the ones presented in \cite{page_DependentModelingTemporal_2022} to the non-Markovian case. Their proofs are reported in Section SM.2 of the Supplementary Materials.

\begin{proposition} \label{prop:rhok}
    The conditional probability reported in Equation~\eqref{eq:rho_prior} can be written as
    \begin{equation*}
        \pr(\rho_k \mid \bm{\gamma}_{k-d_\rho+1:k}, \rho_{k-1}) = \frac{\pr(\rho_k)}{\pr(\rho_k^{\mathcal{R}_k})} \ind{\rho_{k-1}^{\mathcal{R}_k} = \rho_k^{\mathcal{R}_k}} = 
        \begin{cases}
            \pr(\rho_k) / \pr(\rho_k^{\mathcal{R}_k}) & \text{if $\rho_{k-1}^{\mathcal{R}_k} = \rho_k^{\mathcal{R}_k}$} \\
            0 & \text{otherwise}
        \end{cases}.
    \end{equation*}
\end{proposition}

\begin{proposition} \label{prop:Rk}
    Recalling that $\mathcal{R}_{k'}$ are the indices of units that are not reallocated when moving from index $k'-1$ to $k'$ and let $\mathcal{R}_{k'}^{(+i,k)}$ and $\mathcal{R}_{k'}^{(-i,k)}$ be the same quantity with $\gamma_{ik}$ set to 1 and 0, respectively. Then, for $k=k'-d_\rho+1,\ldots,k'$, $k'=1,\ldots,K$,
    \begin{equation*}
    \begin{split}
        \mathcal{R}_{k'}^{(+i,k)} &=
        \begin{cases}
            \mathcal{R}_{k'} \cup \{i\} & \text{if $\max\left\{\bm{\gamma}_{i,k'-d_\rho+1:k'}\right\}=0$} \\
            \mathcal{R}_{k'} & \text{otherwise}
        \end{cases}, \\
        \mathcal{R}_{k'}^{(-i,k)} &=
        \begin{cases}
            \mathcal{R}_{k'} \setminus \{i\} & \text{if $\max\left\{\bm{\gamma}_{i,k'-d_\rho+1:k'}^{(-k)}\right\}=0$ and $\gamma_{ik}=1$ in $\mathcal{R}_{k'}$} \\
            \mathcal{R}_{k'} & \text{otherwise}
        \end{cases}.
    \end{split}
    \end{equation*}
    where $\bm{\gamma}_{i,k'-d_\rho+1:k'}^{(-k)}$ denotes $\bm{\gamma}_{i,k'-d_\rho+1:k'}$ with $\gamma_{ik}$ set to 0.
    
    Moreover, it holds that $\mathcal{R}_{k'}^{(+i,k)}=\mathcal{R}_{k'}^{(-i,k)}\cup\{i\}$ if $\max\left\{\bm{\gamma}_{i,k'-d_\rho+1:k'}^{(-k)}\right\}=0$.
\end{proposition}

From the previous propositions, we can derive the full conditional distribution for $\gamma_{ik}$, which is
\footnotesize
\begin{equation} \label{eq:full_cond_gamma}
    \pr(\gamma_{ik}=1\mid-) = \frac{\alpha_k}{\alpha_k + (1-\alpha_k) \prod_{k'=k}^{k+d_\rho-1} \pr(\rho_{k'}^{\mathcal{R}_{k'}^{(+i,k)}})/\pr(\rho_{k'}^{\mathcal{R}_{k'}^{(-i,k)}})} \prod_{k'=k}^{k+d_\rho-1} \ind{\rho_{k'-1}^{\mathcal{R}_{k'}^{(+i,k)}}=\rho_{k'}^{\mathcal{R}_{k'}^{(+i,k)}}},
\end{equation}
\normalsize
where factors with $k'>K$ are set to one. The quantity $\pr(\rho_{k'}^{\mathcal{R}_{k'}^{(+i,k)}}) / \pr(\rho_{k'}^{\mathcal{R}_{k'}^{(-i,k)}})$ can be easily computed using Proposition~\ref{prop:Rk} and Neal's Algorithm 8 with one auxiliary parameter \citep{neal_MarkovChainSampling_2000}. As a consequence of Proposition~\ref{prop:Rk}, $\pr(\rho_{k'}^{\mathcal{R}_{k'}^{(+i,k)}})/\pr(\rho_{k'}^{\mathcal{R}_{k'}^{(-i,k)}})$ needs to be computed only when updating $\gamma_{ik}$ changes the set of compatible partitions for $\rho_{k'}$, since
\begin{equation*}
    \frac{\pr\left(\rho_{k'}^{\mathcal{R}_{k'}^{(+i,k)}}\right)}{\pr\left(\rho_{k'}^{\mathcal{R}_{k'}^{(-i,k)}}\right)} =
    \begin{cases}
        \pr\left(c_{ik'}\mid \rho_{k'}^{\mathcal{R}_{k'}^{(-i,k)}}\right) & \text{if $\mathcal{R}_{k'}^{(+i,k)}=\mathcal{R}_{k'}^{(-i,k)}\cup\{i\}$} \\
        1 & \text{otherwise}
    \end{cases}.
\end{equation*}
Assuming $\rho_1\sim\CRP(M)$, when $\mathcal{R}_{k'}^{(+i,k)}=\mathcal{R}_{k'}^{(-i,k)}\cup\{i\}$, that is when $i\notin\mathcal{R}_{k'}^{(-i,k)}$, we have
\begin{equation} \label{eq:cik}
    \pr\left(c_{ik'}=j\mid \rho_{k'}^{\mathcal{R}_{k'}^{(-i,k)}}\right) =
    \begin{cases}
        n_j(\rho_{k'}^{\mathcal{R}_{k'}^{(-i,k)}}) \Big/ \left(n(\rho_{k'}^{\mathcal{R}_{k'}^{(-i,k)}}) + M\right)  & \text{if $j=1,\ldots,J(\rho_{k'}^{\mathcal{R}_{k'}^{(-i,k)}})$} \\
        M \Big/ \left(n(\rho_{k'}^{\mathcal{R}_{k'}^{(-i,k)}}) + M\right) & \text{if $j=J(\rho_{k'}^{\mathcal{R}_{k'}^{(-i,k)}})+1$}
    \end{cases},
\end{equation}
where $n_j(\rho_{k'}^{\mathcal{R}_{k'}^{(-i,k)}})$ is the size of cluster $j$ in $\rho_{k'}^{\mathcal{R}_{k'}^{(-i,k)}}$, $n(\rho_{k'}^{\mathcal{R}_{k'}^{(-i,k)}})$ is the number of units in $\rho_{k'}^{\mathcal{R}_{k'}^{(-i,k)}}$ and $J(\rho_{k'}^{\mathcal{R}_{k'}^{(-i,k)}})$ is the number of clusters in $\rho_{k'}^{\mathcal{R}_{k'}^{(-i,k)}}$.

We again use Proposition~\ref{prop:rhok} and Neal's Algorithm 8 with one auxiliary parameter to derive the full conditional distribution for $c_{ik}$, which is
\small
\begin{equation*}
    \pr(c_{ik}=j\mid-) \propto \pr(\rho_k^{(i\to j)}) \pr(\mathbf{Q}\mid c_{ik}=j,\ldots) \ind{\rho_{k-1}^{\mathcal{R}_k} = (\rho_k^{(i\to j)})^{\mathcal{R}_k}} \ind{(\rho_k^{(i\to j)})^{\mathcal{R}_{k+1}} = \rho_{k+1}^{\mathcal{R}_{k+1}}},
\end{equation*}
\normalsize
where $\pr(\mathbf{Q}\mid c_{ik}=j,\ldots)$ is the joint distribution of the random variables $\mathbf{Q}$ whose distributions depend on $c_{ik}$ with $c_{ik}$ set to $j$.
The quantity $\rho_k^{(i\to j)}$ denotes the partition $\rho_k$ with unit $i$ assigned to cluster $j$, where $S^{(-i)}_{kj}$ is the $j$th cluster at index $k$, $S_{kj}$, with unit $i$ removed, and $J_k^{(-i)}$ is the number of clusters in $\rho_k^{(-i)}$, the partition $\rho_k$ with unit $i$ removed; $j=J_k^{(-i)}+1$ corresponds to the new singleton cluster $\{i\}$. Assuming again $\rho_1\sim\CRP(M)$,
\begin{equation*}
    \pr\left(\rho_k^{(i\to j)}\right) \propto
    \begin{cases}
        |S^{(-i)}_{kj}| & \text{if $j=1,\ldots,J_k^{(-i)}$} \\
        M & \text{if $j=J_k^{(-i)}+1$}
    \end{cases},
\end{equation*}
the full conditional distribution for $c_{ik}=j$ becomes
\scriptsize
\begin{equation} \label{eq:full_cond_cik}
    \pr(c_{ik}=j\mid-) \propto
    \begin{cases}
        |S^{(-i)}_{kj}| \pr(\mathbf{Q}\mid c_{ik}=j,\ldots) \ind{\rho_{k-1}^{\mathcal{R}_k} = (\rho_k^{(i\to j)})^{\mathcal{R}_k}} \ind{(\rho_k^{(i\to j)})^{\mathcal{R}_{k+1}} = \rho_{k+1}^{\mathcal{R}_{k+1}}} & \text{if $j=1,\ldots,J_k^{(-i)}$} \\
        M \pr(\mathbf{Q}\mid c_{ik}=j,\ldots) \ind{\rho_{k-1}^{\mathcal{R}_k} = (\rho_k^{(i\to j)})^{\mathcal{R}_k}} \ind{(\rho_k^{(i\to j)})^{\mathcal{R}_{k+1}} = \rho_{k+1}^{\mathcal{R}_{k+1}}} & \text{if $j=J_k^{(-i)}+1$}
    \end{cases},
\end{equation}
\normalsize
where any auxiliary parameters, i.e., variables related to the new cluster, are drawn from their prior distributions as in \citet{neal_MarkovChainSampling_2000}, used to compute $\pr(c_{ik}=J_k^{(-i)}+1\mid-)$, and then kept if the new cluster is selected, discarded otherwise.
The first indicator function implies that the local cluster $c_{ik}$ is updated only if the local cluster can be reallocated when moving from index $k-1$ to $k$ since the indicator function sets $\Pr(c_{ik}=j\mid-)$ to 0 if the partitions $\rho_{k-1}$ and $\rho_k^{(i\to j)}$ are not compatible with respect to $\bm{\gamma}_{k-d_\rho+1:k}$; the second indicator implies that the local cluster $c_{ik}$ cannot be reallocated to another already existing cluster if the local cluster cannot be reallocated when moving from index $k$ to $k+1$ since the indicator function sets $\Pr(c_{ik}=j\mid-)$ to 0 if the partitions $\rho_k^{(i\to j)}$ and $\rho_{k+1}$ are not compatible with respect to $\bm{\gamma}_{k-d_\rho+2:k+1}$.
In practice, the two indicator functions are first evaluated, and the full conditional is then computed only if $c_{ik}$ can be allocated to a different cluster.
Notably, for $k=1$, the first indicator function is omitted since $\rho_1$ has no preceding partition; for $k=K$, the second indicator function is omitted since $\rho_{K+1}$ does not exist.

Finally, the full conditional distribution for $\alpha_k$, for $k=2,\ldots,K$, is
\begin{equation} \label{eq:full_cond_alpha}
    \alpha_k \mid- \sim \Beta\left( a_\alpha+\sum_{i=1}^n\gamma_{ik}, b_\alpha+n-\sum_{i=1}^n\gamma_{ik}\right).
\end{equation}

A general Gibbs sampler, which can be used for any model employing smRPM, is reported in Algorithm \ref{alg:gibbs}.
Although our algorithm may appear to be a straightforward extension of the tRPM sampler, several additional considerations are required to accommodate the increased complexity of the latent variable updates in \eqref{eq:full_cond_gamma} and \eqref{eq:full_cond_cik}.
The full conditional distribution of $c_{ik}$ depends on the dependence order $d_\rho$ only through the sets $\mathcal{R}_k$ and $\mathcal{R}_{k+1}$, namely the indices of units that are not reallocated when moving from index $k-1$ to $k$, and from $k$ to $k+1$, which are required to evaluate the two indicator functions. These sets are computed once per MCMC iteration, after updating the auxiliary variables.
On the other hand, the $d_\rho$-order dependence introduces the product of $d_\rho$ indicator functions in the numerator and $d_\rho$ predictive rules in the denominator of the full conditional distribution for $\gamma_{ik}$.
In practice, these quantities are rarely computed in full: the predictive rules are computed only when all indicators are satisfied since $\gamma_{ik}$ is set to zero as soon as any of the $d_\rho$ indicator functions is not satisfied.
Moreover, a predictive rule is effectively evaluated through Neal's Algorithm 8 only if updating $\gamma_{ik}$ modifies the set of compatible partitions.
As a result, although the number of quantities appearing in the full conditionals grows with $d_\rho$, the effective number of performed evaluations is smaller.

\begin{algorithm}[t]
    \caption{One iteration of the Gibbs sampler for a model employing smRPM}
    \label{alg:gibbs}
    \begin{enumerate}
        \item[1.] Sample $\gamma_{ik}$, for $k=2,\ldots,K$ and $i=1,\ldots,n$, using \eqref{eq:full_cond_gamma}.
        \item[2.] Sample $c_{ik}$, for $k=1,\ldots,K$ and $i=1,\ldots,n$, using \eqref{eq:full_cond_cik}.
        \item[3.] Sample $\alpha_k$, for $k=2,\ldots,K$, using \eqref{eq:full_cond_alpha}.
        \item[4.] Sample remaining model-specific parameters.
    \end{enumerate}
\end{algorithm}

\section{Simulation study}
\label{sec:simulations}
We assess the performance of smRPM through three simulation studies.  In the first one, we consider time series data and assess the model performance as a function of $d_\rho$. In the second and third, instead, we consider functional data. In this case we match $d_\rho$ with the degree of the B-spline basis. The first two simulations provide guidelines to choose $d_\rho$ which are applied in the third simulation where the method is compared with existing alternatives in the motivating task of functional local clustering problem.

In all studies with functional data, we use B-splines with equally spaced knots.
We define a $n^{(ref)}\times K$ matrix of cluster assignments $\mathbf{C}^{(ref)}$ for $n^{(ref)}$ reference time series or functional observations, and build the $n\times K$ matrix of cluster assignments $\mathbf{C}$ for the simulated dataset by replicating $n^{(rep)}$ times each row of $\mathbf{C}^{(ref)}$.
By doing so, we obtain simulated datasets containing $n^{(rep)}$ realizations of each of these reference time series or functional observations. We set $n^{(ref)}=5$ and let $n^{(rep)}$ vary over the set $\{10,30\}$, thus obtaining two scenarios with 50 and 150 observations, respectively. Data are generated with high and low signal-to-noise ratios. 

To assess the ability of each model in retrieving the true sequence of partitions, we computed the Adjusted Rand Index \citep[ARI, ][]{hubert_ComparingPartitions_1985} between the true partition at each index $k$, $\rho_k$, and the visited partition at iteration $b$ of the MCMC, $\hat{\rho}^{(b)}_k$.
Then, we computed a posterior estimate of ARI,  as
\begin{equation*}
    \text{ARI}({\rho_k},\hat{\rho}_k^{(\mathcal{B})}) = \frac{1}{|\mathcal{B}|}\sum_{b\in\mathcal{B}} \ARI(\rho_k,\hat{\rho}_k^{(b)}),
\end{equation*}
where $\hat{\rho}_k^{(\mathcal{B})}$ collects all the $|\mathcal{B}|$ partitions $\hat{\rho}^{(b)}_k$ sampled by the MCMC algorithm, and $\mathcal{B}$ is a set of 1000 MCMC samples, obtained by collecting 10000 samples, discarding the first 5000 and thinning by 5, in each model considered during the simulation studies.
The results presented henceforth are based on a single MCMC run. Visual inspection of several  chains revealed no apparent convergence or mixing issues. The sampler is initialized by assigning all observations to a single cluster.
Details on the prior hyperparameter choice and initialization of cluster-specific parameters are reported in Section SM.5 of the Supplementary Materials.

In the first simulation, data are simulated according to
\begin{equation*}
    Y_{ik} \mid \bm{\mu}^*_k, \sigma^{*2}, c_{ik} \indsim \Normal(\mu^*_{k,c_{ik}},\sigma^{*2}), \quad i=1,\ldots,n,\quad k=1,\ldots,K,
\end{equation*}
with $\sigma^{*2} \in  \{1,4\}$. Unlike the functional setting, in which $d_\rho$ coincides with the degree of the B-spline, $d_\rho$ needs to be tuned or selected according to previous knowledge. To evaluate the performance of the model under correct and incorrect specification of $d_\rho$, we consider two sub-scenarios with different $\bm{\mu}^*$ and $\mathbf{C}$, where the sequence of partitions $\mathbf{C}$ exhibits Markovian dependence in the first sub-scenario and second-order dependence in the second one. 
We fit the model reported in Section SM.3 of the Supplementary Materials for different values of $d_\rho\in\{1,2,3\}$. Notably, tRPM of \citet{page_DependentModelingTemporal_2022} is a special case with $d_\rho=1$ and is reported for comparison.

\begin{figure}
    \centering
    \includegraphics[width=0.9\linewidth]{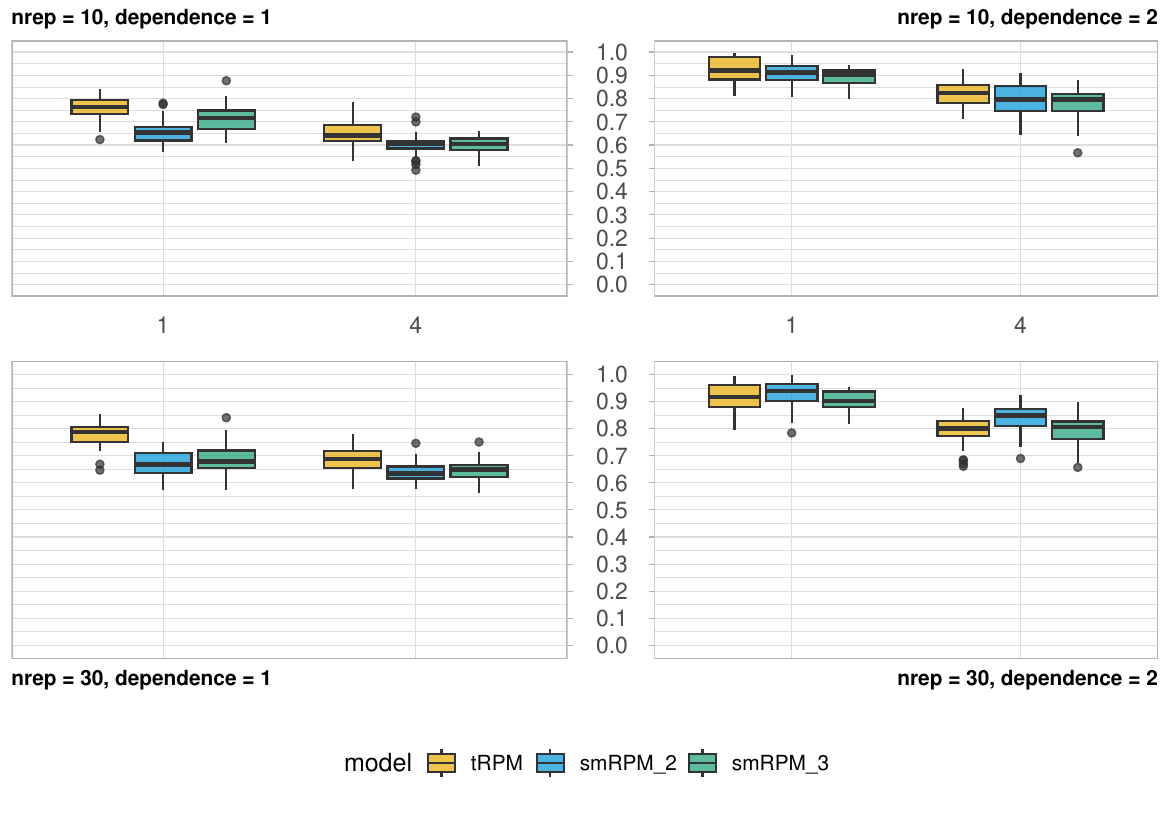}
    \caption{Boxplot of the average posterior ARI across the sequence of partitions computed on $R=50$ simulated datasets under Markovian (left) and second-order dependence (right). Each figure considers a different combination of number of replications $n^{(rep)}$ and dependence order: from top left clockwise, $(10,1)$, $(10,2)$, $(30,2)$ and $(30,1)$. Each figure shows the low-noise scenario on the left ($\sigma^{*2}=1$) and the high-noise one on the right ($\sigma^{*2}=4$).}
    \label{fig:boxplot_ARI}
\end{figure}

Figure~\ref{fig:boxplot_ARI} reports the results. As expected, the model employing tRPM performs best when the sequence of partitions exhibits Markovian dependence. Under second-order dependence, the correctly specified model with $d_\rho=2$ achieves the best performance in the larger sample size scenario, while the models perform comparably with smaller sample sizes.
The noise introduced by $\sigma^{*2}$ influences the models' performance while not changing their ranking.  Additional details, including time-specific posterior ARIs, are reported in Section SM.5 of the Supplementary Materials.

As a second simulation experiment, we simulated data under the model introduced in \eqref{eq:model_theta}-\eqref{eq:model_curves}. All curves are assumed to be observed on a common equispaced grid $(x_1,\ldots,x_T)$ of $T=100$ evaluation points over the interval $[0,1]$. 
We considered cubic B-splines with $18$ evenly spaced knots between 0 and 1, corresponding to $K=20$ basis functions, to simulate the data.
Details are reported in Section SM.5 of the Supplementary Materials.

We fit the model with $d_\rho\in\{1,3\}$ representing the tRPM of \citet{page_DependentModelingTemporal_2022} and the proposed smRPM.
The dependence order $d_\rho$ matches the degree of the cubic B-splines to fully exploit the local property.
To assess sensitivity to the number of basis functions in the expansion, each model is fitted three times: once using the same B-spline basis employed to generate the data, and twice using misspecified B-splines with $K-5$ and $K+5$ basis functions, respectively.

Figure~\ref{fig:boxplot_fARI} reports the results based on the functional version of ARI (yARI), as defined in Section SM.5 of the Supplementary Materials. As expected, models whose partition dependence aligns with the degree of the B-spline outperform the models utilizing the tRPM.
The choice of the B-spline basis for estimation has an impact on performance, with results suggesting that overestimating the number of basis functions is preferable to underestimating it.
Similarly to the non-functional case, the noise parameter $\sigma^{2}$ and the number of observations influence the model performance but do not affect their ranking.

\begin{figure}
    \centering
    \includegraphics[width=0.9\linewidth]{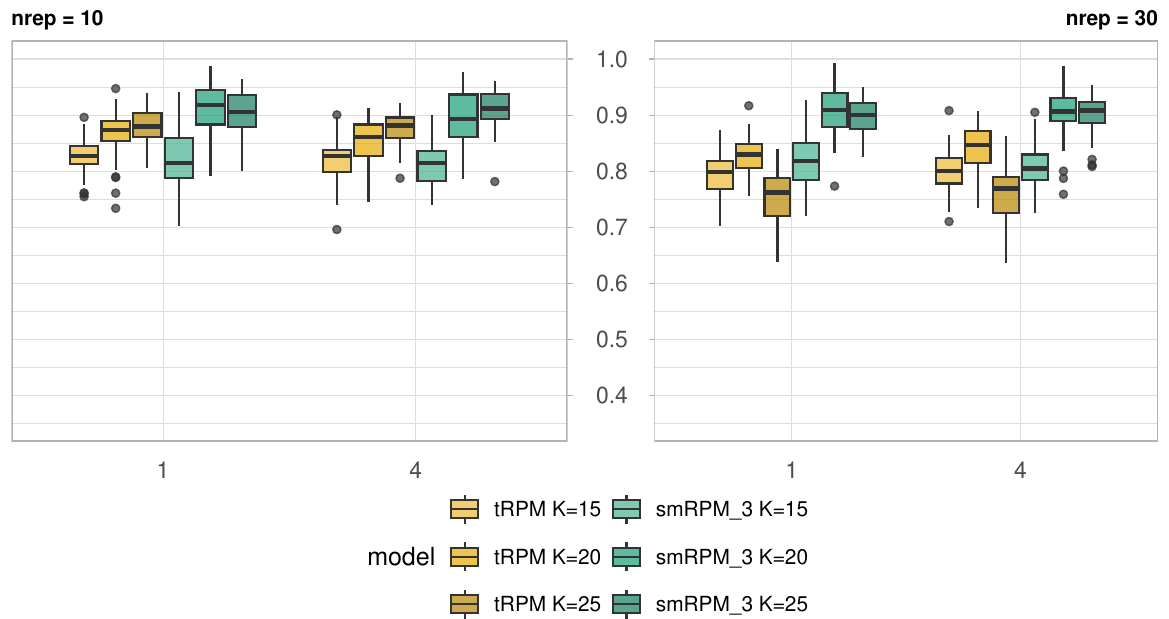}
    \caption{Boxplot of the average posterior yARI across the sequence of functional partitions computed on $R=50$ simulated datasets. The figure on the left considers $n^{(rep)}=10$, while the one on the right $n^{(rep)}=30$. Each figure shows the low-noise scenario on the left ($\sigma^{2}=1$) and the high-noise one on the right ($\sigma^{2}=4$).}
    \label{fig:boxplot_fARI}
\end{figure}

Finally, we assess the empirical performance of the proposed approach with $d_\rho$ set equal to the degree of the B-spline basis, comparing it to other competing methods. To the best of our knowledge, the only alternative implementation currently available, apart from tRPM, is the Hidden Markov Functional Local Clustering Model (HMFLCM) introduced by \citet{fan_BayesianSemiparametricLocal_2024}. The latter is a two-step procedure that performs global clustering in the first step and local clustering on the global clusters in the second step that exploits the same local property of B-splines of the proposed functional smRPM.

To evaluate the competing approaches under both correct and incorrect specifications, we consider scenarios where the B-splines in the data-generating mechanism either match or differ from those of the models.
Under correct specification, we both simulate the data and fit the model using quadratic B-splines with $K=20$ basis functions and knots coinciding with the $T=19$ evaluation points.
Under misspecification, we fit the proposed approach using B-splines of degrees 1, 2, 3, 4, and 5 with $K=20$ basis functions, which is substantially fewer than the 80 basis functions of the 4-degree B-splines employed in the data-generating mechanism. In this second scenario, all curves are assumed to be observed on a common equispaced grid of $T=25$ evaluation points over the interval $[0,1]$.
This allows us to evaluate the sensitivity to the spline degree and to demonstrate that satisfactory performance can be achieved even with a limited number of basis functions and evaluation points. Details are reported in Section SM.5 of the Supplementary Materials.

Figure~\ref{fig:boxplot_fARI_RMSE_HMFLCM_s2} reports the results under misspecification. 
In addition to yARI, we also computed a posterior estimate of Root Mean Square Error (RMSE) to assess the ability of each model in predicting observations. Details on how the RMSE is computed are reported in Section SM.5 of the Supplementary Materials.
In terms of yARI, the dependence order $d_\rho$, which always matches the degree of the B-splines, strongly influences the performance of the proposed approach, with functional smRPM using $d_\rho>1$ always performing the best.
HMFLCM shows the poorest and best performance in the smaller and larger sample size scenario, respectively.
In terms of RMSE, the functional smRPMs using $d_\rho>1$ achieve comparable results, followed by functional tRPM and lastly HMFLCM, in the smaller sample size scenario, while higher dependence orders correspond to lower RMSE with larger sample sizes.
\begin{figure}
    \centering
    \includegraphics[width=0.95\linewidth]{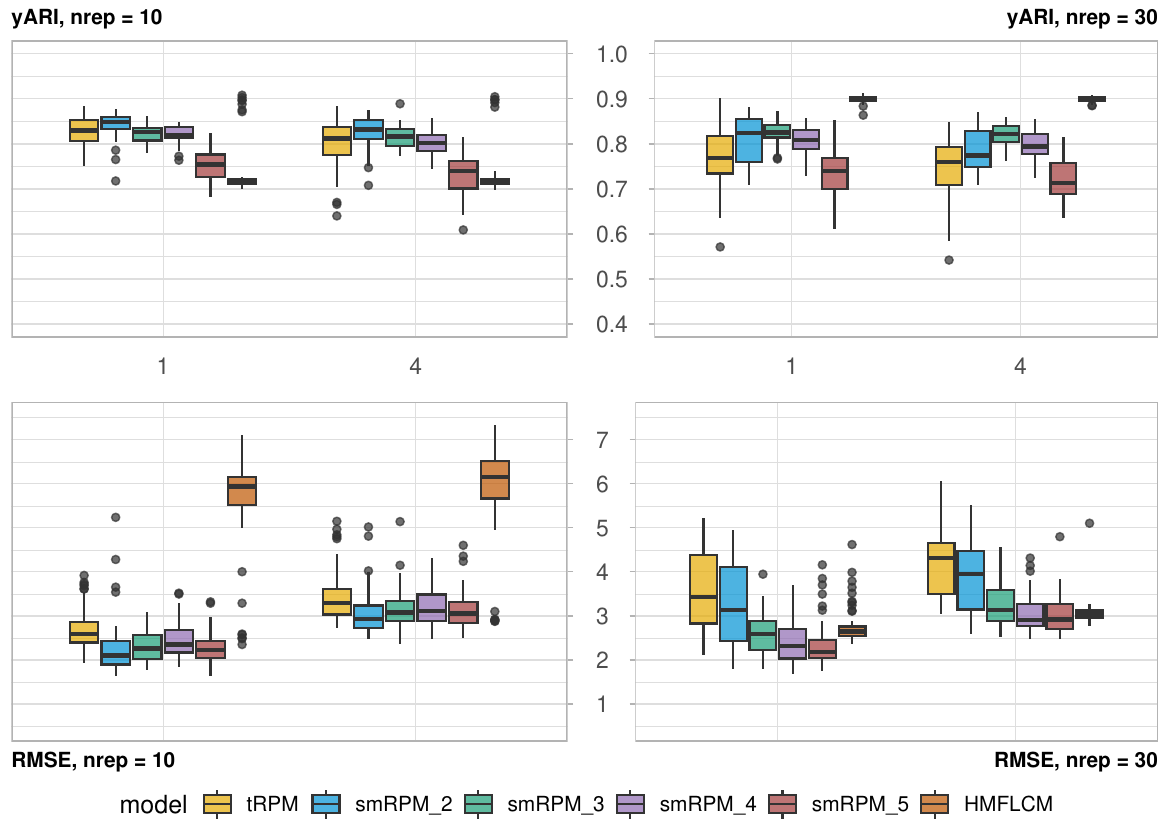}
    \caption{Boxplot of the average posterior yARI across the sequence of functional observations (top) and posterior RMSE between functional observations and predicted counterparts (bottom) computed on $R=50$ simulated datasets under misspecification of the B-splines. The figures on the left consider $n^{(rep)}=10$, while the ones on the right $n^{(rep)}=30$. Each figure shows the low-noise setting on the left ($\sigma^{2}=1$) and the high-noise one on the right ($\sigma^{2}=4$).}
    \label{fig:boxplot_fARI_RMSE_HMFLCM_s2}
\end{figure}

A similar figure for the correctly specified scenario is reported in Section SM.5 of the Supplementary Materials.
All models show comparable results in terms of fARI, while HMFLCM has a much higher RMSE compared to the other models, which once again show similar performance.

\section{Venetian Lagoon tide level data analysis}
\label{sec:venice}
In this section, we apply our proposed model for functional data to a  dataset of tidal level measurements from monitoring stations in the Venetian Lagoon. 
The Venetian Lagoon is an ecosystem in perpetual transformation whose morphology is shaped by the constant interaction between the tides and the erosion caused by waves. Venice, located within this delicate environment, faces significant flood risks that cause disruptions in mobility, difficulties in many public services, issues in numerous ground-floor establishments, and gradual deterioration of building structures.
This issue has intensified due to a combination of natural phenomena, e.g., sea-level rise, and human interventions that have altered the lagoon environment.
Over the years, several measures have been implemented to counter these disruptive transformations, among which the most important is MOSE, a large-scale flood protection system consisting of rows of mobile barriers that, if activated, can stop the rise of the tides by isolating the lagoon from the sea during high tides.
Therefore, effective and continuous monitoring of water levels is crucial for flood management and for operating the MOSE system.

We consider hourly tidal level data, measured at $n=11$ monitoring stations in the Venetian Lagoon, that are publicly available and easily downloadable from the official website of the City of Venice.
In particular, we focus on two time intervals: a shorter one from August 27 to August 30, 2023 ($T=95$ evaluation points), and a longer one from October 19 to October 26, 2023 ($T=192$).
During some of these days, the MOSE system was activated to block rising tides and prevent flooding within the lagoon.
Before applying the model, functional data have been registered, as customary in functional data analysis \citep{ramsay_FunctionalDataAnalysis_2005}, with a simple informed approach. Specifically, we shifted  the evaluation points of certain stations by 0.5, 1.0 or 1.5 hours depending on their distance from the sea, reflecting the time lags in tidal rise reaching different stations.

For both datasets we fit a functional smRPM with cubic B-splines.
We considered 61 evenly spaced knots for the longer interval and 36 for the shorter one, resulting in  $K=63$ and $K=38$ basis functions, respectively.
Detailed description of the prior hyperparameters is provided in Section SM.6 of the Supplementary Materials.
A posterior point estimate for the sequence of partitions $\rho_1, \ldots, \rho_K$ is obtained by independently estimating each partition $\rho_k$ at every index $k$.
Each of these is estimated using the SALSO algorithm \citep{dahl_SearchAlgorithmsLoss_2022} with the Binder loss \citep{binder_BayesianClusterAnalysis_1978}. 
As alternative loss, one could use the Variation of Information \citep[VI,][]{meila_ComparingClusteringsInformation_2007,wade_BayesianClusterAnalysis_2018}. 
In this application, despite the former being known for identifying more clusters than the latter, we obtained the same posterior point estimate using both losses.
Posterior estimates of the cluster-specific B-spline parameters, $\bm{\theta}^*$, are computed conditionally on the posterior point estimate of $\rho_1, \ldots, \rho_K$; see Section SM.4.4 of the Supplementary Materials for the detailed procedure.

We compare the observed time series of tide level with the curves estimated by our model in Figure~\ref{fig:venice_october} and Figure~\ref{fig:venice_august}.
Along with the estimated curves, the bottom panel shows the posterior distribution of the number of local clusters on each subinterval $\mathcal{D}_k$.

The longer time interval depicted in Figure~\ref{fig:venice_october} covers eight days in October 2023, during which the MOSE system was activated. Specifically, it was operational during the first two days, 19 and 20, as well as the last three days, 24, 25 and 26.
It is interesting to note that all the stations have the same behavior when the MOSE system is not active, while several different local behaviors arise when it is operational. We intentionally did not account for this to showcase the ability of the model to detect both global and local clusters in an unsupervised fashion. 
The large number of clusters on October 24th can be explained by the severe storms on that day, which made the measurements from the stations less precise and more variable from station to station.

\begin{figure}
    \centering
    \includegraphics[width=\linewidth]{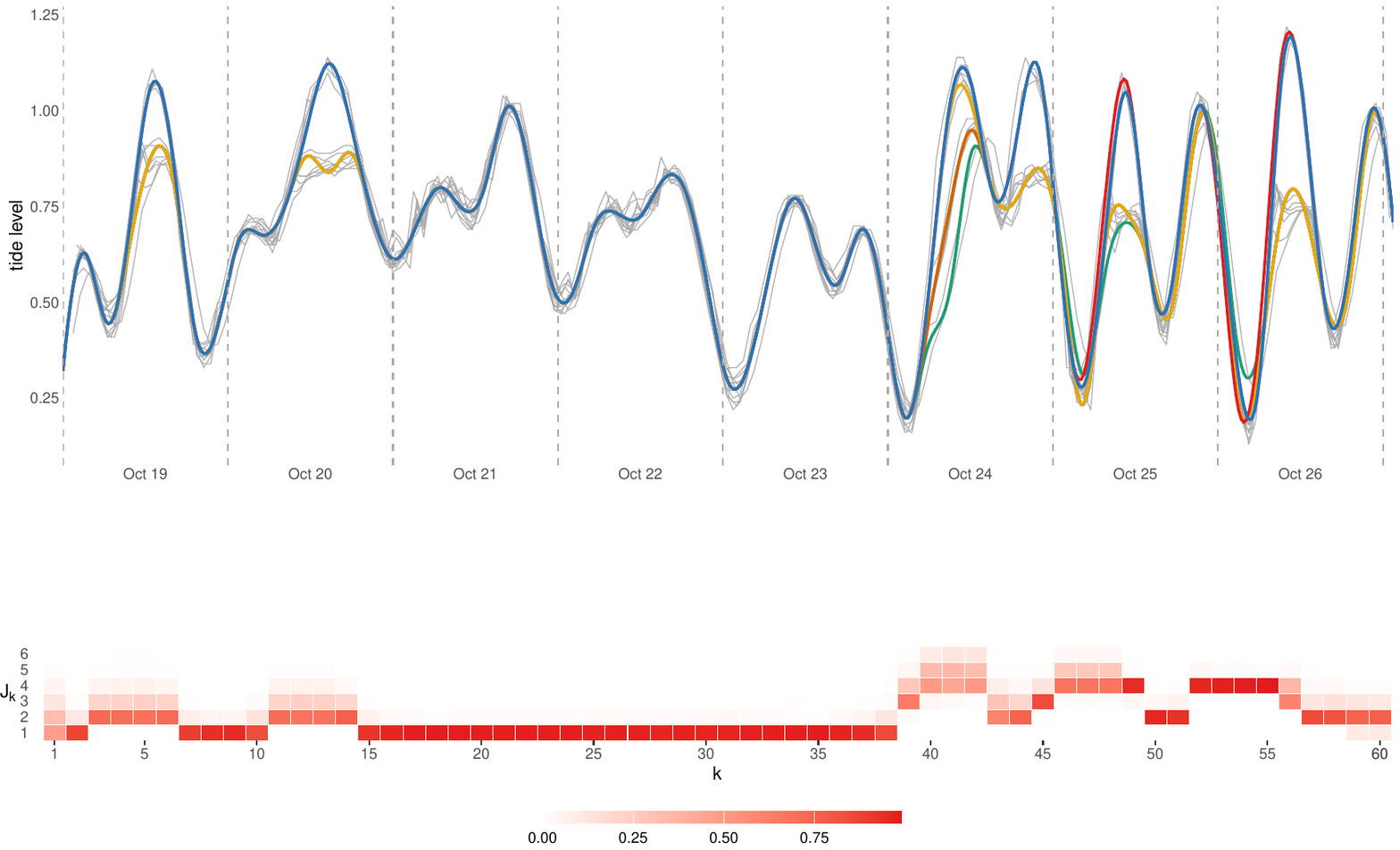}
    \caption{Top: comparison of the observed tide levels (grey) and the corresponding estimated curves for the time interval October 19-26, 2023. Bottom: posterior distribution of the number of local clusters on each subinterval $\mathcal{D}_k$; the posterior probability is represented with a color scale ranging from white for low probability to red for higher probability.}
    \label{fig:venice_october}
\end{figure}

Similarly, the shorter interval depicted in Figure~\ref{fig:venice_august} provides a clearer visualization of how local clustering facilitates a more nuanced analysis of functional data compared to methods relying solely on global clustering.
In this case, all stations belong to the blue cluster throughout the interval without merging or splitting into any other group, except for the yellow cluster, which emerges briefly on August 28th and 29th, days marked by exceptionally high tides for the summer period in which the MOSE system was activated. The blue cluster can be interpreted as representing global behaviors, while still allowing for the identification of local patterns, represented by the splitting of the yellow curve.
Notably, the posterior distribution of the number of clusters provides further insights. While the point estimate suggests two clusters, the posterior distribution extends beyond this, capturing the uncertainty in the clustering process. This feature, customary in Bayesian model-based clustering, complements the point estimate by providing a probabilistic assessment of the number of clusters and their stability.

\begin{figure}
    \centering
    \includegraphics[width=\linewidth]{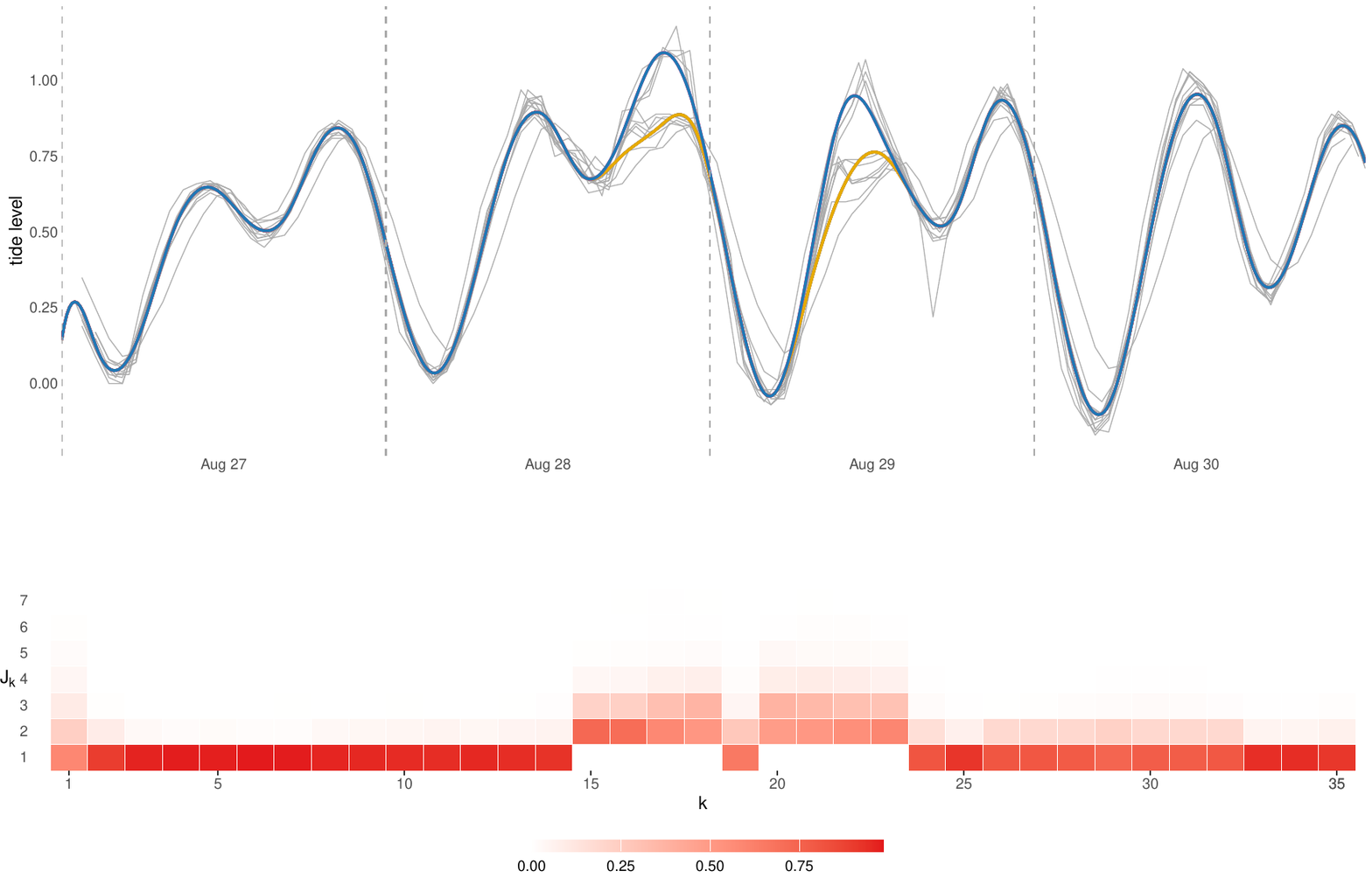}
    \caption{Top: comparison of the observed tide levels (grey) and the corresponding estimated curves for the time interval August 27-30, 2023. Bottom: posterior distribution of the number of local clusters on each subinterval $\mathcal{D}_k$; the posterior probability is represented with a color scale ranging from white for low probability to red for higher probability.}
    \label{fig:venice_august}
\end{figure}

\section{Conclusion}
\label{sec:conclusion}
In this article, we proposed a Bayesian hierarchical approach to perform indirect local clustering of functional data through B-spline basis expansion and a dependent RPM for the basis coefficients. The model exploits the local property of B-splines to obtain partially coincident functions when sharing clusters for enough contiguous basis parameters.
We proposed a novel dependent RPM, smRPM, that relaxes the assumptions of tRPM by assuming $d_\rho$-order dependence in the distribution of the partitions.
Through simulation studies, we showed that the model for functional data is able to fully exploit the local property of B-splines and tends to perform best when the dependence order is set equal to the degree of the B-spline employed in the model, although performance may vary across scenarios.
The proposed approach offers a flexible alternative to first-order Markovian specifications and can, in principle, be adapted to other settings where higher-order dependence is of interest.

This work can be extended in several directions, both related to the model for functional data and to the novel dependent RPM.
The Gaussian noise in \eqref{eq:model_curves} could be replaced with an alternative formulation which takes into account the dependence across the domain of the functional data, e.g., Gaussian process. 
The relatively simple approach to link contiguous cluster-specific B-spline parameters could be extended to jointly model groups of parameters, possibly borrowing ideas from Bayesian P-splines \citep{lang_BayesianPSplines_2004}.
Moving our attention to smRPM, the distribution of the auxiliary variables could be easily modified to also handle curve-specific, basis-specific and curve-basis-specific covariates in a similar fashion to \cite{liang_BayesianNonparametricApproach_2024}.
Last but not least, alternative EPPFs could be employed to add more flexibility to smRPM and eventually include additional prior information on the clustering. \citet{page_DependentModelingTemporal_2022} considered spatial Product Partition Model \citep{page_SpatialProductPartition_2016}, however one can use any other RPM for modeling the first partition and use smRPM mechanism to handle the temporal evolution. Alternatively, any EPPF induced by an element of the class of Gibbs-type priors \citep{deblasi_AreGibbsTypePriors_2015}, among which we recall the Pitman-Yor Process \citep{pitman_TwoparameterPoissonDirichletDistribution_1997}, is a good choice since it allows for a closed form computation of $\pr(\rho_{k'}^{\mathcal{R}_{k'}^{(+i,k)}})/\pr(\rho_{k'}^{\mathcal{R}_{k'}^{(-i,k)}})$ in \eqref{eq:full_cond_gamma}.

\subsubsection*{Supplementary Materials}
The supplementary material provides details on posterior computation and additional results for the empirical illustrations. The source code for implementation is available at \url{https://github.com/giovannitoto/smRPM_code}.

\bibliographystyle{agsm}
\bibliography{Bibliography}


\newpage
\setcounter{page}{1}
\renewcommand{\thepage}{\texorpdfstring{\arabic{page}}{SM.\arabic{page}}}

\spacingset{1}
\begin{center}
{\large\bf SUPPLEMENTARY MATERIAL}
\end{center}

\spacingset{1.4} 

\appendix

\setcounter{section}{0}
\setcounter{table}{0}
\setcounter{figure}{0}
\setcounter{equation}{0}
\renewcommand{\thefigure}{SM\arabic{figure}}
\renewcommand{\thesection}{SM.\arabic{section}}
\renewcommand{\thetable}{SM\arabic{table}}
\renewcommand{\theequation}{SM.\arabic{equation}}
\renewcommand{\theHfigure}{SM.\arabic{figure}}
\renewcommand{\theHtable}{SM.\arabic{table}}
\renewcommand{\theHequation}{SM.\arabic{equation}}

\makeatletter
\def\hyper@natlinkstart#1{\Hy@backout{#1}\hyper@linkstart{cite}{cite.SM.#1}\def\hyper@nat@current{#1}}
\def\hyper@natlinkbreak#1#2{\hyper@linkend#1\hyper@linkstart{cite}{cite.SM.#2}}
\def\hyper@natanchorstart#1{\Hy@raisedlink{\hyper@anchorstart{cite.SM.#1}}}
\makeatother

\begin{bibunit}[agsm]
\addtocontents{toc}{\protect\setcounter{tocdepth}{2}}

\cftsetindents{section}{0em}{4em}
\cftsetindents{subsection}{1em}{4em}
\tableofcontents

\newpage 

\section{Notation}

In this section we review and summarize the notation and the acronyms adopted in the paper. 

\subsubsection*{Sequence of partitions}

Let $\rho_{1},\ldots,\rho_{K}$ represent a sequence of  evolving partitions with  $\rho_k=\{S_{k1},\ldots,S_{kJ_k}\}\in\mathcal{P}$ denoting the partition at domain index $k$ of the $n$ objects into $J_k$ clusters, $S_{kj}\subseteq\{1,\ldots,n\}$, $k=1,\ldots,K$; let $\mathcal{P}$ denote the set of all possible partitions.

In this paper, we consider the following Random Partition Models (RPMs) as prior distribution for one or more random partitions:

{\scriptsize\vspace{0.5cm}
\begin{tabular}{lcc}
    \hline
    Random Partition Model  &  & Parameters\\
    \hline
    Chinese Restaurant Process & $\CRP(M)$ & concentration parameter $M>0$ \\
    temporal RPM & $\tRPM(\bm{\alpha}, M)$ & $\bm{\alpha}\in(0,1)^K$, concentration parameter $M>0$ \\
    semi-Markovian RPM & $\smRPM_{d_\rho}(\bm{\alpha},M)$ & $\bm{\alpha}\in(0,1)^K$, concentration parameter $M>0$ \\
    \hline
\end{tabular}\vspace{0.5cm}}

\subsubsection*{Auxiliary variables}
Under smRPM, the evolution of the partitions at different $k=1,\ldots,K$ is guided by auxiliary variables.
For each unit $i$ and domain index $k$, we introduce a binary latent variable $\gamma_{ik}$ specifying whether the $i$th unit cannot be considered for possible cluster reallocation for the next $d_\rho$ consecutive index transitions, i.e.,
\begin{equation*}
    \gamma_{ik} =
    \begin{cases}
        1 & \text{if unit $i$ \emph{cannot} be reallocated in the $d_\rho$  consecutive indices after $k-1$} \\
        0 & \text{otherwise}
    \end{cases}.
\end{equation*}
The $n$ auxiliary variables at index $k$, denoted as $\bm{\gamma}_k=(\gamma_{1k},\ldots,\gamma_{nk})$,  influence the partitions at the next $d_\rho$ indices.
This also implies that $\rho_k$ depends on $\bm{\gamma}_{k-d_\rho+1},\ldots,\bm{\gamma}_{k}$; analogously, the $i$th unit in $\rho_k$, $c_{ik}$, depends on $\gamma_{i,k-d_\rho+1},\ldots,\gamma_{i,k}$.
Inspired by \texttt{R}, we denote the two sequences by $\bm{\gamma}_{k-d_\rho+1:k}$ and $\bm{\gamma}_{i,k-d_\rho+1:k}$.

\subsubsection*{Probability distributions}
We consider the following probability distributions:

{\vspace{0.5cm}
\begin{tabular}{lll}
    \hline
    Name & Notation & Parameters\\
    \hline
    Bernoulli & $\Ber(\pi)$ & probability of success $\pi\in(0,1)$ \\
    Beta & $\Beta(a,b)$ & $a >0$, $b>0$ \\
    Gamma & $\Ga(a,b)$ & shape $a>0$, rate $b>0$ \\
    Inverse Gamma & $\InvGa(a,b)$ & shape $a>0$, scale $b>0$ \\
    Normal & $\Normal(\mu,\sigma^2)$  & mean $\mu\in\R$, variance $\sigma^2>0$ \\
    \hline
\end{tabular}\vspace{0.5cm}}

\noindent We denote the probability density function (pdf) of a random variable evaluated at $x$ as $\text{Dist}(x;\text{parameters})$. For example, $\Normal(x;\mu,\sigma^2)$ is the pdf of a $\Normal(\mu,\sigma^2)$ evaluated at $x$.

The notation $\pr(\cdot)$ is used as a general notation for the probability distribution of a generic random variable. Let $X$ and $Y$ be random variables, then

{\vspace{0.5cm}
\begin{tabular}{rl}
    $\pr(X)$ & is the probability distribution of $X$, \\
    $\pr(X,Y)$ & is the joint probability distribution of $X$ and $Y$, \\
    $\pr(X\mid Y)$ & is the conditional probability distribution of $X$ given $Y$, \\
    $\pr(X\mid-)$ & is the conditional probability distribution of $X$ given all other variables.
\end{tabular}\vspace{0.5cm}}

\section{Proofs of Section 3}

\begin{proof}[Proof of Proposition \ref{prop:rhok}]
    We want to prove that
    \begin{equation*}
        \frac{\pr(\rho_k)\ind{\rho_k\in \mathcal{P}_{\mathcal{R}_k}}}{\sum_{\lambda'\in \mathcal{P}}\pr(\rho_k=\lambda')\ind{\lambda'\in \mathcal{P}_{\mathcal{R}_k}}} = \frac{\pr(\rho_k)}{\pr(\rho_k^{\mathcal{R}_k})} \ind{\rho_{k-1}^{\mathcal{R}_k} = \rho_k^{\mathcal{R}_k}}.
    \end{equation*}
    Both indicator functions in the numerator check whether $\rho_{k-1}$ and $\rho_k$ are compatible with respect to $\bm{\gamma}_{k-d_\rho+1:k}$, so in practice they are checking the same statement written in different forms:
    \begin{equation*}
        \rho_k\in \mathcal{P}_{\mathcal{R}_k} = \{\rho_k\in \mathcal{P}:\rho_{k-1}^{\mathcal{R}_k}=\rho_{k}^{\mathcal{R}_k}\} \quad\iff\quad \rho_{k-1}^{\mathcal{R}_k} = \rho_k^{\mathcal{R}_k}.
    \end{equation*}
    Now we only need to prove that $\sum_{\lambda'\in \mathcal{P}}\pr(\rho_k=\lambda')\ind{\lambda'\in \mathcal{P}_{\mathcal{R}_k}} = \pr(\rho_k^{\mathcal{R}_k})$:
    \begin{equation*}
        \sum_{\lambda'\in \mathcal{P}}\pr(\rho_k=\lambda')\ind{\lambda'\in \mathcal{P}_{\mathcal{R}_k}} = \sum_{\lambda'\in \mathcal{P}_{\mathcal{R}_k}}\pr(\rho_k=\lambda') = \sum_{c_{ik}:i\notin\mathcal{R}_k}\pr(\rho_k=\mathbf{c}_{\cdot k}) = \pr\left(\rho_k^{\mathcal{R}_k}\right).
    \end{equation*}
\end{proof}

\begin{proof}[Proof of Proposition \ref{prop:Rk}]
    First of all, we recall that $i\notin\mathcal{R}_{k'}$, i.e., unit $i$ can be reallocated when moving from index $k'-1$ to $k'$, if $\max\left\{\bm{\gamma}_{i,k'-d_\rho+1:k'}\right\}=0$, and $i\in\mathcal{R}_{k'}$, i.e., unit $i$ cannot be reallocated when moving from index $k'-1$ to $k'$, if $\max\left\{\bm{\gamma}_{i,k'-d_\rho+1:k'}\right\}=1$.
    It is worth noting that setting $\gamma_{ik}$ to 1 may or may not add unit $i$ to the set $\mathcal{R}_{k'}$ depending on the other values in $\bm{\gamma}_{i,k'-d_\rho+1:k'}$, that are collected in $\bm{\gamma}_{i,k'-d_\rho+1:k'}^{(-k)}$. Similarly, setting $\gamma_{ik}$ to 0 may or may not remove unit $i$ from the set $\mathcal{R}_{k'}$ depending on the values in $\bm{\gamma}_{i,k'-d_\rho+1:k'}^{(-k)}$. 

    We want to prove
    \begin{equation*}
        \mathcal{R}_{k'}^{(+i,k)} =
        \begin{cases}
            \mathcal{R}_{k'} \cup \{i\} & \text{if $\max\left\{\bm{\gamma}_{i,k'-d_\rho+1:k'}\right\}=0$} \\
            \mathcal{R}_{k'} & \text{otherwise}
        \end{cases}.
    \end{equation*}
    The set $\mathcal{R}_{k'}^{(+i,k)}$ always contains unit $i$ due to $\gamma_{ik}$ set to 1, however $\mathcal{R}_{k'}$ may or may not contain unit $i$ depending on the values in $\bm{\gamma}_{i,k'-d_\rho+1:k'}$.
    We have two cases:
    \begin{enumerate}
        \item if all values in $\bm{\gamma}_{i,k'-d_\rho+1:k'}$ are 0, i.e., $\max\left\{\bm{\gamma}_{i,k'-d_\rho+1:k'}\right\}=0$, then unit $i$ can be reallocated when moving from index $k'-1$ to $k'$ and therefore $i\notin\mathcal{R}_{k'}$;
        \item if there is at least one 1 in $\bm{\gamma}_{i,k'-d_\rho+1:k'}$, i.e., $\max\left\{\bm{\gamma}_{i,k'-d_\rho+1:k'}\right\}=1$, then unit $i$ cannot be reallocated when moving from index $k'-1$ to $k'$ and therefore $i\in\mathcal{R}_{k'}$.
    \end{enumerate}
    So, setting $\gamma_{ik}$ to 1 in case 1. adds unit $i$ to the set $\mathcal{R}_{k'}$, leading to $\mathcal{R}_{k'}^{(+i,k)}=\mathcal{R}_{k'} \cup \{i\}$ if $\max\left\{\bm{\gamma}_{i,k'-d_\rho+1:k'}\right\}=0$; while setting $\gamma_{ik}$ to 1 leaves the set unchanged in case 2. since the set already contains unit $i$, leading to $\mathcal{R}_{k'}^{(+i,k)} =\mathcal{R}_{k'}$ if $\max\left\{\bm{\gamma}_{i,k'-d_\rho+1:k'}\right\}=1$.

    We want to prove
    \begin{equation*}
        \mathcal{R}_{k'}^{(-i,k)} =
        \begin{cases}
            \mathcal{R}_{k'} \setminus \{i\} & \text{if $\max\left\{\bm{\gamma}_{i,k'-d_\rho+1:k'}^{(-k)}\right\}=0$ and $\gamma_{ik}=1$ in $\mathcal{R}_{k'}$} \\
            \mathcal{R}_{k'} & \text{otherwise}
        \end{cases}.
    \end{equation*}
    The set $\mathcal{R}_{k'}$ contains unit $i$ if there is at least one 1 in $\bm{\gamma}_{i,k'-d_\rho+1:k'}$; analogously, the set $\mathcal{R}_{k'}^{(-i,k)}$ contains unit $i$ if there is at least one 1 in $\bm{\gamma}_{i,k'-d_\rho+1:k'}^{(-k)}$. We have three cases:
    \begin{enumerate}
        \item if all values in $\bm{\gamma}_{i,k'-d_\rho+1:k'}^{(-k)}$ are 0, i.e., $\max\left\{\bm{\gamma}_{i,k'-d_\rho+1:k'}^{(-k)}\right\}=0$, and $\gamma_{ik}=1$ in $\mathcal{R}_{k'}$, then unit $i$ cannot be reallocated when moving from index $k'-1$ to $k'$, i.e., $i\in\mathcal{R}_{k'}$, only because of $\gamma_{ik}$ being equal to 1;
        \item if there is at least one 1 in $\bm{\gamma}_{i,k'-d_\rho+1:k'}^{(-k)}$, then unit $i$ cannot be reallocated when moving from index $k'-1$ to $k'$, and therefore $i\in\mathcal{R}_{k'}^{(-i,k)}$, regardless of the value assumed by $\gamma_{ik}$;
        \item if all values in $\bm{\gamma}_{i,k'-d_\rho+1:k'}$ are zero, setting $\gamma_{ik}$ to 0 does not change the sets since $\gamma_{ik}$ is already 0.
    \end{enumerate}
    So, setting $\gamma_{ik}$ to 0 in case 1. removes unit $i$ from the set $\mathcal{R}_{k'}$, leading to $\mathcal{R}_{k'}^{(-i,k)} = \mathcal{R}_{k'} \setminus \{i\}$ if $\max\left\{\bm{\gamma}_{i,k'-d_\rho+1:k'}^{(-k)}\right\}=0$ and $\gamma_{ik}=1$ in $\mathcal{R}_{k'}$; while setting $\gamma_{ik}$ to 0 in the other two cases leaves the set unchanged since both sets contain unit $i$ due to $\bm{\gamma}_{i,k'-d_\rho+1:k'}^{(-k)}$ containing at least one 1 in case 2. and both sets do not contain unit $i$ due to the absence of 1 in $\bm{\gamma}_{i,k'-d_\rho+1:k'}$ in case 3..
    
    We want to prove that $\mathcal{R}_{k'}^{(+i,k)}=\mathcal{R}_{k'}^{(-i,k)}\cup\{i\}$ if $\max\left\{\bm{\gamma}_{i,k'-d_\rho+1:k'}^{(-k)}\right\}=0$. Although it is not part of the proposition, we also prove that
    \begin{equation*}
        \frac{\pr\left(\rho_{k'}^{\mathcal{R}_{k'}^{(+i,k)}}\right)}{\pr\left(\rho_{k'}^{\mathcal{R}_{k'}^{(-i,k)}}\right)} =
        \begin{cases}
            \pr\left(c_{ik'}\mid \rho_{k'}^{\mathcal{R}_{k'}^{(-i,k)}}\right) & \text{if $\mathcal{R}_{k'}^{(+i,k)}=\mathcal{R}_{k'}^{(-i,k)}\cup\{i\}$} \\
            1 & \text{otherwise}
        \end{cases}.
    \end{equation*}
    The equality $\mathcal{R}_{k'}^{(+i,k)}=\mathcal{R}_{k'}^{(-i,k)}\cup\{i\}$ holds only when $\mathcal{R}_{k'}^{(-i,k)}$ does not contain unit $i$, that is when all values in $\bm{\gamma}_{i,k'-d_\rho+1:k'}^{(-k)}$ are 0, i.e., $\max\left\{\bm{\gamma}_{i,k'-d_\rho+1:k'}^{(-k)}\right\}=0$.  
    If $\mathcal{R}_{k'}^{(+i,k)}=\mathcal{R}_{k'}^{(-i,k)}\cup\{i\}$ holds, then the partition $\rho_{k'}^{\mathcal{R}_{k'}^{(+i,k)}}$ is the partition $\rho_{k'}^{\mathcal{R}_{k'}^{(-i,k)}}$ with the additional unit $i$, allocated to cluster $c_{ik'}$ and, consequently,
    \begin{equation*}
        \frac{\pr\left(\rho_{k'}^{\mathcal{R}_{k'}^{(+i,k)}}\right)}{\pr\left(\rho_{k'}^{\mathcal{R}_{k'}^{(-i,k)}}\right)} =
        \frac{\pr\left(c_{ik'},\rho_{k'}^{\mathcal{R}_{k'}^{(-i,k)}}\right)}{\pr\left(\rho_{k'}^{\mathcal{R}_{k'}^{(-i,k)}}\right)} =
        \frac{\pr\left(c_{ik'}\mid\rho_{k'}^{\mathcal{R}_{k'}^{(-i,k)}}\right)\pr\left(\rho_{k'}^{\mathcal{R}_{k'}^{(-i,k)}}\right)}{\pr\left(\rho_{k'}^{\mathcal{R}_{k'}^{(-i,k)}}\right)} = \pr\left(c_{ik'}\mid\rho_{k'}^{\mathcal{R}_{k'}^{(-i,k)}}\right).
    \end{equation*}
    Obviously, $\pr(\rho_{k'}^{\mathcal{R}_{k'}^{(+i,k)}})=\pr(\rho_{k'}^{\mathcal{R}_{k'}^{(-i,k)}})$ if $\mathcal{R}_{k'}^{(+i,k)}=\mathcal{R}_{k'}^{(-i,k)}$ since we are evaluating the EPPF at the same partition.
\end{proof}

\section{Models employing smRPM} \label{SM-section:model_based_on_smRPM}
We propose two models employing the semi-Markovian Random Partition Model (smRPM). We consider a model for non-functional data and one for functional data to show that the usage of smRPM is not limited to our motivating functional data problem, but it can be employed in any context in which we are interested in the direct modeling of a sequence of dependent partitions.
Before that, we recall the joint probabilistic structure of smRPM and the prior distribution for $\bm{\alpha}$:
\begin{equation*}
\begin{split}
    \rho_1 &\sim \CRP(M), \\
    \pr(\rho_k=\lambda \mid \bm{\gamma}_{k-d_\rho+1:k}, \rho_{k-1}) &= \frac{\pr(\rho_k=\lambda)\ind{\lambda\in \mathcal{P}_{\mathcal{R}_k}}}{\sum_{\lambda'\in \mathcal{P}}\pr(\rho_k=\lambda')\ind{\lambda'\in \mathcal{P}_{\mathcal{R}_k}}}, \\
    \gamma_{ik} \mid \alpha_k &\indsim \Ber(\alpha_k), \\
    \alpha_k &\iidsim \Beta(a_\alpha,b_\alpha).
\end{split}
\end{equation*}

\subsection{Time series data model}
Before moving to the more complex case with functional data, we consider a more intuitive case in which $n$ time series with $K$ observations each are modeled. In particular, we consider a hierarchical model, similar to the one proposed in \citet{page_DependentModelingTemporal_2022}, where time appears in the partition model only. The main difference between the two models is the use of smRPM instead of tRPM, and the replacement of the uniform prior distributions with inverse Gamma distributions. The last change was made to remove the Metropolis steps in the Gibbs sampler.
Using cluster label notation, the model for non-functional data is
\begin{equation} \label{SM-eq:model_nonfunctional}
\begin{split}
    Y_{ik} &\mid \bm{\mu}^*_k, \bm{\sigma}^{*2}_k, c_{ik} \indsim \Normal(\mu^*_{k,c_{ik}},\sigma^{*2}_{k,c_{ik}}), \quad i=1,\ldots,n, \quad k=1,\ldots,K \\
    \mu^*_{kj} &\mid \theta_k, \tau^2_k \indsim \Normal(\theta_k,\tau^2_k), \quad
    \sigma^{*2}_{kj} \iidsim \InvGa(a_\sigma,b_\sigma), \quad j=1,\ldots,J_k,\quad k=1,\ldots,K \\
    \theta_k &\mid \phi_0,\lambda^2 \indsim \Normal(\phi_0,\lambda^2), \quad
    \tau^2_k \iidsim \InvGa(a_\tau,b_\tau), \quad k=1,\ldots,K \\
    \phi_0 &\sim \Normal(m_0,s_0^2), \quad
    \lambda^2 \sim \InvGa(a_\lambda,b_\lambda), \\
    \mathbf{C} &\mid \bm{\alpha} \sim \smRPM_{d_\rho}(\bm{\alpha},M), \quad \alpha_k \iidsim \Beta(a_\alpha,b_\alpha),
\end{split} 
\end{equation}
where $\bm{\mu}^*_k=(\mu^*_{k1},\ldots,\mu^*_{kJ_k})^\top\in\R^{J_k}$ and $\bm{\sigma}^{*2}_k=(\sigma^{*2}_{k1},\ldots,\sigma^{*2}_{kJ_k})^\top\in(0,\infty)^{J_k}$.
Posterior inference is performed via MCMC. The full conditional distributions are reported in the next section.

\subsection{Functional data model}
Using cluster label notation, the model for functional data described in the main article is 
\begin{equation} \label{SM-eq:model_functional}
\begin{split}
    Y_i(x_{it}) \mid \bm{\theta}^*,\mathbf{c}_{i},\sigma^2 &\indsim \Normal(\mathbf{b}(x_{it})^\top\bm{\theta}_i,\sigma^2), \quad t=1,\ldots,T_i, \quad i=1,\ldots,n, \\
    \theta^*_{kj} \mid \bm{\theta}^*_{k-1},\mathbf{c}_{\cdot,k},\mathbf{c}_{\cdot,k-1},\phi,\tau^2 &\indsim \Normal\left(\frac{\phi}{|\mathcal{C}_{k-1}^{(\to j)}|}\sum_{l\in \mathcal{C}_{k-1}^{(\to j)}}\theta^*_{k-1,l},\tau^2\right), \quad j=1,\ldots, J_k, \quad k=2,\ldots,K, \\
    \theta^*_{1j} \mid \tau^2 &\indsim \Normal(0, \tau^2), \quad j=1,\ldots, J_1, \\
    \phi &\sim \Normal(m_0,s^2_0), \\
    \tau^2 &\sim \InvGa(a_\tau,b_\tau), \\
    \sigma^2 &\sim \InvGa(a_\sigma,b_\sigma), \\
    \mathbf{C} \mid \bm{\alpha} &\sim \smRPM_{d_\rho}(\bm{\alpha},M), \quad \alpha_k \iidsim \Beta(a_\alpha,b_\alpha),
\end{split} 
\end{equation}
where $\bm{\theta}_i =(\theta^*_{1,c_{i1}},\ldots,\theta^*_{K,c_{iK}})^\top\in\R^K$.
Posterior inference is performed via MCMC. The full conditional distributions are reported in the next section.

\section{Details on posterior inference via MCMC}
In this section, we report the computations to derive the full conditional distributions for $\gamma_{ik}$ and $c_{ik}$, $i=1,\ldots,n$, $k=1,\ldots,K$, under smRPM, and report the full conditional distributions for all latent quantities in the model for time series and the model for functional data.

\subsection{smRPM}
To derive the full conditional distribution for $\gamma_{ik}$ and $c_{ik}$, $i=1,\ldots,n$, $k=1,\ldots,K$, we focus on the joint conditional distribution of $\bm{\gamma}$, $\mathbf{C}$, $\bm{\alpha}$ and the random variables $\mathbf{Q}$ whose distributions depend on $\mathbf{C}$. Notice that we will not specify $\mathbf{Q}$ and its distribution in this section since here we are focusing on smRPM only, but $\mathbf{Q}$ depends on the model employing smRPM.
Recalling $\pr(\rho_k) = \pr(c_{1k},\ldots,c_{nk})$, and analogously $\pr(\rho_1,\ldots,\rho_K) = \pr(\mathbf{C})$, in what follows we will use
$\mathbf{C}$ to refer to the sequence of partitions, $\rho_k$ to refer to a single partition and $c_{ik}$ to refer to the cluster allocation of a specific unit at index $k$.

The joint distribution is

\small
\begin{align*}
    \pr(\mathbf{Q},\bm{\gamma},\mathbf{C},\bm{\alpha}) =& \pr(\mathbf{Q}\mid\mathbf{C},\ldots)  \pr(\mathbf{C},\bm{\gamma}\mid\bm{\alpha}) \pr(\bm{\alpha}) \\
    =& \pr(\mathbf{Q}\mid\mathbf{C},\ldots) \pr(\rho_1,\ldots,\bm{\gamma}_K,\rho_K\mid\bm{\alpha}) \pr(\bm{\alpha}) \\
    =& \pr(\mathbf{Q}\mid\mathbf{C},\ldots) \pr(\rho_1) \prod_{k=2}^K \left\{ \pr(\rho_k \mid \bm{\gamma}_{k-d_\rho+1:k}, \rho_{k-1}) \pr(\bm{\gamma}_k\mid \alpha_k) \pr(\alpha_k)\right\} \\
    =& \pr(\mathbf{Q}\mid\mathbf{C},\ldots) \pr(\rho_1) \prod_{k=2}^K \left\{ \pr(\rho_k \mid \bm{\gamma}_{k-d_\rho+1:k}, \rho_{k-1}) \pr(\alpha_k) \prod_{i=1}^n \pr(\gamma_{ik} \mid \alpha_k) \right\}.
\end{align*}
\normalsize

\subsubsection*{Full conditional distribution for $\gamma_{ik}$}
The full conditional for $\gamma_{ik}$ is
\begin{align*}
    \pr(\gamma_{ik}\mid-) &\propto \pr(\gamma_{ik} \mid \alpha_k) \times \prod_{k'=k}^{k+d_\rho-1}\pr(\rho_{k'} \mid \bm{\gamma}_{k'-d_\rho+1:k'}, \rho_{k'-1}) \\
    &= \alpha_k^{\gamma_{ik}}(1-\alpha_k)^{1-\gamma_{ik}} \times \prod_{k'=k}^{k+d_\rho-1} \frac{\pr(\rho_{k'})}{\pr(\rho_{k'}^{\mathcal{R}_{k'}})} \ind{\rho_{k'-1}^{\mathcal{R}_{k'}}=\rho_{k'}^{\mathcal{R}_{k'}}}.
\end{align*}
The full conditional distribution for $\gamma_{ik}=1$ is
\begin{equation*}
    \pr(\gamma_{ik}=1\mid-) \propto \alpha_k \prod_{k'=k}^{k+d_\rho-1} \frac{\pr(\rho_{k'})}{\pr(\rho_{k'}^{\mathcal{R}_{k'}^{(+i,k)}})} \ind{\rho_{k'-1}^{\mathcal{R}_{k'}^{(+i,k)}}=\rho_{k'}^{\mathcal{R}_{k'}^{(+i,k)}}}
\end{equation*}
and, similarly, the full conditional distribution for $\gamma_{ik}=0$ is
\begin{equation*}
    \pr(\gamma_{ik}=0\mid-) \propto (1 - \alpha_k) \prod_{k'=k}^{k+d_\rho-1} \frac{\pr(\rho_{k'})}{\pr(\rho_{k'}^{\mathcal{R}_{k'}^{(-i,k)}})} \ind{\rho_{k'-1}^{\mathcal{R}_{k'}^{(-i,k)}}=\rho_{k'}^{\mathcal{R}_{k'}^{(-i,k)}}}.
\end{equation*}
Normalizing, we obtain
\begin{equation} \label{SM-eq:full_cond_gamma}
    \pr(\gamma_{ik}=1\mid-) = \frac{\alpha_k}{\alpha_k + (1 - \alpha_k) \prod_{k'=k}^{k+d_\rho-1} \pr(\rho_{k'}^{\mathcal{R}_{k'}^{(+i,k)}})/\pr(\rho_{k'}^{\mathcal{R}_{k'}^{(-i,k)}})} \prod_{k'=k}^{k+d_\rho-1} \ind{\rho_{k'-1}^{\mathcal{R}_{k'}^{(+i,k)}}=\rho_{k'}^{\mathcal{R}_{k'}^{(+i,k)}}},
\end{equation}
where factors with $k'>K$ are set to one, and $\pr(\rho_{k'}^{\mathcal{R}_{k'}^{(+i,k)}})/\pr(\rho_{k'}^{\mathcal{R}_{k'}^{(-i,k)}})$ can be easily computed using Proposition~\ref{prop:Rk} and Neal's Algorithm 8 with one auxiliary parameter \citep{neal_MarkovChainSampling_2000}. As a consequence of Proposition~\ref{prop:Rk}, $\pr(\rho_{k'}^{\mathcal{R}_{k'}^{(+i,k)}})/\pr(\rho_{k'}^{\mathcal{R}_{k'}^{(-i,k)}})$ needs to be computed only when updating $\gamma_{ik}$ changes the set of compatible partitions for $\rho_{k'}$ since
\begin{equation*}
    \frac{\pr\left(\rho_{k'}^{\mathcal{R}_{k'}^{(+i,k)}}\right)}{\pr\left(\rho_{k'}^{\mathcal{R}_{k'}^{(-i,k)}}\right)} =
    \begin{cases}
        \pr\left(c_{ik'}\mid \rho_{k'}^{\mathcal{R}_{k'}^{(-i,k)}}\right) & \text{if $\mathcal{R}_{k'}^{(+i,k)}=\mathcal{R}_{k'}^{(-i,k)}\cup\{i\}$} \\
        1 & \text{otherwise}
    \end{cases}.
\end{equation*}
Assuming $\rho_1\sim\CRP(M)$, when $\mathcal{R}_{k'}^{(+i,k)}=\mathcal{R}_{k'}^{(-i,k)}\cup\{i\}$, that is when $i\notin\mathcal{R}_{k'}^{(-i,k)}$, we have
\begin{equation} \label{SM-eq:cik}
    \pr\left(c_{ik'}=j\mid \rho_{k'}^{\mathcal{R}_{k'}^{(-i,k)}}\right) =
    \begin{cases}
        n_j(\rho_{k'}^{\mathcal{R}_{k'}^{(-i,k)}}) \Big/ \left(n(\rho_{k'}^{\mathcal{R}_{k'}^{(-i,k)}}) + M\right)  & \text{if $j=1,\ldots,J(\rho_{k'}^{\mathcal{R}_{k'}^{(-i,k)}})$} \\
        M \Big/ \left(n(\rho_{k'}^{\mathcal{R}_{k'}^{(-i,k)}}) + M\right) & \text{if $j=J(\rho_{k'}^{\mathcal{R}_{k'}^{(-i,k)}})+1$}
    \end{cases},
\end{equation}
where $n_j(\rho_{k'}^{\mathcal{R}_{k'}^{(-i,k)}})$ is the size of cluster $j$ in $\rho_{k'}^{\mathcal{R}_{k'}^{(-i,k)}}$, $n(\rho_{k'}^{\mathcal{R}_{k'}^{(-i,k)}})$ is the number of units in $\rho_{k'}^{\mathcal{R}_{k'}^{(-i,k)}}$ and $J(\rho_{k'}^{\mathcal{R}_{k'}^{(-i,k)}})$ is the number of clusters in $\rho_{k'}^{\mathcal{R}_{k'}^{(-i,k)}}$.

\subsubsection*{Full conditional distribution for $c_{ik}$}
We again use Proposition~\ref{prop:rhok} and Neal's Algorithm 8 with one auxiliary parameter to derive the full conditional distribution for $c_{ik}$, which is

\footnotesize
\begin{equation*}
\begin{split}
    \pr(c_{ik}=j\mid-) \propto& \pr(\rho_k^{(i\to j)} \mid \bm{\gamma}_{k-d_\rho+1:k}, \rho_{k-1}) \pr(\rho_{k+1} \mid \bm{\gamma}_{k-d_\rho+2:k+1}, \rho_k^{(i\to j)}) \pr(\mathbf{Q}\mid c_{ik}=j,\ldots) \\
    =& \frac{\pr(\rho_k^{(i\to j)})}{\pr((\rho_k^{(i\to j)})^{\mathcal{R}_k})} \ind{\rho_{k-1}^{\mathcal{R}_k} = (\rho_k^{(i\to j)})^{\mathcal{R}_k}} \frac{\pr(\rho_{k+1})}{\pr(\rho_{k+1}^{\mathcal{R}_{k+1}})} \ind{(\rho_k^{(i\to j)})^{\mathcal{R}_{k+1}} = \rho_{k+1}^{\mathcal{R}_{k+1}}} \pr(\mathbf{Q}\mid c_{ik}=j,\ldots) \\
    \propto& \pr(\rho_k^{(i\to j)}) \pr(\mathbf{Q}\mid c_{ik}=j,\ldots) \ind{\rho_{k-1}^{\mathcal{R}_k} = (\rho_k^{(i\to j)})^{\mathcal{R}_k}} \ind{(\rho_k^{(i\to j)})^{\mathcal{R}_{k+1}} = \rho_{k+1}^{\mathcal{R}_{k+1}}},
\end{split}
\end{equation*}
\normalsize
where $\pr(\mathbf{Q}\mid c_{ik}=j,\ldots)$ is the joint distribution of the random variables whose distributions depend on $c_{ik}$ with $c_{ik}$ set to $j$.
The quantity $\rho_k^{(i\to j)}$ denotes the partition $\rho_k$ with unit $i$ assigned to cluster $j$, where $S^{(-i)}_{kj}$ is the $j$th cluster at index $k$, $S_{kj}$, with unit $i$ removed, and $J_k^{(-i)}$ is the number of clusters in $\rho_k^{(-i)}$, the partition $\rho_k$ with unit $i$ removed; $j=J_k^{(-i)}+1$ corresponds to the new singleton cluster $\{i\}$. Assuming again $\rho_1\sim\CRP(M)$,
\begin{equation*}
    \pr\left(\rho_k^{(i\to j)}\right) \propto
    \begin{cases}
        M \Gamma(|S^{(-i)}_{kj}\cup\{i\}|) \prod_{j'\neq j}^{J_k^{(-i)}} M \Gamma(|S^{(-i)}_{kj'}|) & \text{if $j=1,\ldots,J_k^{(-i)}$} \\
        M \Gamma(|\{i\}|) \prod_{j'=1}^{J_k^{(-i)}} M \Gamma(|S^{(-i)}_{kj'}|) & \text{if $j=J_k^{(-i)}+1$}
    \end{cases}.
\end{equation*}
Recalling that $\Gamma(n+1)=n\Gamma(n)$ and $\Gamma(n)=(n-1)!$ for $n\in\N$, the last equation simplifies to
\begin{equation*}
    \pr\left(\rho_k^{(i\to j)}\right) \propto
    \begin{cases}
        |S^{(-i)}_{kj}| & \text{if $j=1,\ldots,J_k^{(-i)}$} \\
        M & \text{if $j=J_k^{(-i)}+1$}
    \end{cases}.
\end{equation*}
The first indicator function, $\ind{\rho_{k-1}^{\mathcal{R}_k} = (\rho_k^{(i\to j)})^{\mathcal{R}_k}}$, implies that the local cluster $c_{ik}$ is updated only if the local cluster can be reallocated when moving from index $k-1$ to $k$ since the indicator function sets $\Pr(c_{ik}=j\mid-)$ to 0 if the partitions $\rho_{k-1}$ and $\rho_k^{(i\to j)}$ are not compatible with respect to $\bm{\gamma}_{k-d_\rho+1:k}$; the second indicator function, $\ind{(\rho_k^{(i\to j)})^{\mathcal{R}_{k+1}} = \rho_{k+1}^{\mathcal{R}_{k+1}}}$, implies that the local cluster $c_{ik}$ cannot be reallocated to another already existing cluster if the local cluster cannot be reallocated when moving from index $k$ to $k+1$ since the indicator function sets $\Pr(c_{ik}=j\mid-)$ to 0 if the partitions $\rho_k^{(i\to j)}$ and $\rho_{k+1}$ are not compatible with respect to $\bm{\gamma}_{k-d_\rho+2:k+1}$.
Notice that $\pr(\rho_k^{\mathcal{R}_k})$ can be omitted since it depends on $c_{ik}$ only when unit $i$ cannot be reallocated when moving from index $k-1$ to $k$.

So, the full conditional distribution for $c_{ik}=j$ is

\footnotesize
\begin{equation} \label{SM-eq:full_cond_cik}
    \pr(c_{ik}=j\mid-) \propto
    \begin{cases}
        |S^{(-i)}_{kj}| \pr(\mathbf{Q}\mid c_{ik}=j,\ldots) \ind{\rho_{k-1}^{\mathcal{R}_k} = (\rho_k^{(i\to j)})^{\mathcal{R}_k}} \ind{(\rho_k^{(i\to j)})^{\mathcal{R}_{k+1}} = \rho_{k+1}^{\mathcal{R}_{k+1}}} & \text{if $j=1,\ldots,J_k^{(-i)}$} \\
        M \pr(\mathbf{Q}\mid c_{ik}=j,\ldots) \ind{\rho_{k-1}^{\mathcal{R}_k} = (\rho_k^{(i\to j)})^{\mathcal{R}_k}} \ind{(\rho_k^{(i\to j)})^{\mathcal{R}_{k+1}} = \rho_{k+1}^{\mathcal{R}_{k+1}}} & \text{if $j=J_k^{(-i)}+1$}
    \end{cases}.
\end{equation}
\normalsize
Any auxiliary parameters, i.e., variables related to the new cluster, are drawn from their prior distributions as in \citet{neal_MarkovChainSampling_2000}, used to compute $\pr(c_{ik}=J_k^{(-i)}+1\mid-)$ and then kept if the new cluster is selected, discarded otherwise.
Notably, for $k=1$, the first indicator function is omitted since $\rho_1$ has no preceding partition; for $k=K$, the second indicator function is omitted since $\rho_{K+1}$ does not exist.

\subsubsection*{Full conditional distributions for $\alpha_k$}
Finally, we report the full conditional distribution for $\alpha_k$, for $k=2,\ldots,K$,
\begin{equation} \label{SM-eq:full_cond_alpha}
    \alpha_k \mid- \sim \Beta\left( a_\alpha+\sum_{i=1}^n\gamma_{ik}, b_\alpha+n-\sum_{i=1}^n\gamma_{ik}\right).
\end{equation}

\subsection{Time series data model}
The full conditional distribution for $\gamma_{ik}$ is Equation~\ref{SM-eq:full_cond_gamma}  (Equation~\ref{eq:full_cond_gamma} in the article), while the full conditional probability for $c_{ik}$ reported in Equation~\ref{SM-eq:full_cond_cik} (Equation~\ref{eq:full_cond_cik} in the article) assumes the following form

\footnotesize
\begin{equation*}
    \pr(c_{ik}=j\mid-) \propto
    \begin{cases}
        |S^{(-i)}_{kj}| \Normal(y_{ik}; \mu^*_{kj}, \sigma^{*2}_{kj}) \ind{\rho_{k-1}^{\mathcal{R}_k} = (\rho_k^{(i\to j)})^{\mathcal{R}_k}} \ind{(\rho_k^{(i\to j)})^{\mathcal{R}_{k+1}} = \rho_{k+1}^{\mathcal{R}_{k+1}}} & \text{if $j=1,\ldots,J_k^{(-i)}$} \\
        M \Normal(y_{ik}; \mu^*_{kj}, \sigma^{*2}_{kj}) \ind{\rho_{k-1}^{\mathcal{R}_k} = (\rho_k^{(i\to j)})^{\mathcal{R}_k}} \ind{(\rho_k^{(i\to j)})^{\mathcal{R}_{k+1}} = \rho_{k+1}^{\mathcal{R}_{k+1}}} & \text{if $j=J_k^{(-i)}+1$}
    \end{cases},
\end{equation*}
\normalsize
where $\mu^*_{k,J_k^{(-i)}+1}$ and $\sigma^{*2}_{k,J_k^{(-i)}+1}$ are drawn from their prior distributions, $\Normal(\theta_k,\tau^2_k)$ and $\InvGa(a_\sigma,b_\sigma)$, respectively.
We further use the full conditional distribution for the elements in $\bm{\alpha}$, reported in Equation~\ref{SM-eq:full_cond_alpha} (Equation~\ref{eq:full_cond_alpha} in the article).
The prior distributions of the remaining parameters are conjugate, hence deriving the full conditional distributions is straightforward. Using standard techniques, we obtain

\footnotesize
\begin{align*}
    \phi_0 \mid - &\sim \Normal\left( \bigg(m_0s_0^{-2}+\lambda^{-2}\sum_{k=1}^K\theta_k\bigg) \bigg(s_0^{-2}+K\lambda^{-2}\bigg)^{-1}, \bigg(s_0^{-2}+K\lambda^{-2}\bigg)^{-1} \right), \\
    \lambda^{-2} \mid - &\sim \Ga\left( a_\lambda+\frac{K}{2}, b_\lambda+\frac{1}{2}\sum_{k=1}^K(\theta_k-\phi_0)^2 \right), \\
    \theta_k \mid - &\sim \Normal\left( \bigg(\phi_0\lambda^{-2}+\tau_k^{-2}\sum_{j=1}^{J_k}\mu^*_{kj}\bigg)\bigg(\lambda^{-2}+J_k\tau_k^{-2}\bigg)^{-1}, \bigg(\lambda^{-2}+J_k\tau_k^{-2}\bigg)^{-1} \right), \\
    \tau_k^{-2} \mid - &\sim \Ga\left( a_\tau+\frac{J_k}{2}, b_\tau+\frac{1}{2}\sum_{j=1}^{J_k}(\mu^*_{kj}-\theta_k)^2 \right), \\
    \mu^*_{kj} \mid - &\sim \Normal\left( \bigg(\theta_k\tau_k^{-2}+\sigma_{kj}^{*-2}\sum_{i:c_{ik}=j}y_{ik}\bigg) \bigg(\tau_k^{-2}+|\{i:c_{ik}=j\}|\sigma_{kj}^{*-2}\bigg)^{-1}, \bigg(\tau_k^{-2}+|\{i:c_{ik}=j\}|\sigma_{kj}^{*-2}\bigg)^{-1} \right), \\
    \sigma_{kj}^{*-2} \mid - &\sim \Ga\left( a_\sigma+\frac{|\{i:c_{ik}=j\}|}{2}, b_\sigma+\frac{1}{2}\sum_{i:c_{ik}=j}(y_{ik}-\mu^*_{kj})^2 \right),
\end{align*}
\normalsize
where $\Ga(a,b)$ is a Gamma distribution with shape $a>0$ and rate $b>0$.

\subsection{Functional data model}
The full conditional distribution for $\gamma_{ik}$ is Equation~\ref{SM-eq:full_cond_gamma}  (Equation~\ref{eq:full_cond_gamma} in the article), while the full conditional probability for $c_{ik}$ reported in Equation~\ref{SM-eq:full_cond_cik} (Equation~\ref{eq:full_cond_cik} in the article) assumes the following form

\tiny
\begin{equation*}
    \pr(c_{ik}=j\mid-) \propto
    \begin{cases}
        |S^{(-i)}_{kj}| \prod_{t=1}^{T_i}\Normal(y_i(x_{it}); \mathbf{b}(x_{it})^\top\bm{\theta}_i, \sigma^2) \ind{\rho_{k-1}^{\mathcal{R}_k} = (\rho_k^{(i\to j)})^{\mathcal{R}_k}} \ind{(\rho_k^{(i\to j)})^{\mathcal{R}_{k+1}} = \rho_{k+1}^{\mathcal{R}_{k+1}}} & \text{if $j=1,\ldots,J_k^{(-i)}$} \\
        M \prod_{t=1}^{T_i}\Normal(y_i(x_{it}); \mathbf{b}(x_{it})^\top\bm{\theta}_i, \sigma^2) \ind{\rho_{k-1}^{\mathcal{R}_k} = (\rho_k^{(i\to j)})^{\mathcal{R}_k}} \ind{(\rho_k^{(i\to j)})^{\mathcal{R}_{k+1}} = \rho_{k+1}^{\mathcal{R}_{k+1}}} & \text{if $j=J_k^{(-i)}+1$}
    \end{cases}.
\end{equation*}
\normalsize
where the auxiliary parameter $\theta^*_{k,J_k^{(-i)}+1}$ is drawn from its prior distribution, $\Normal\left(\phi\theta^{*(\to j)}_{k-1},\tau^2\right)$, where $\theta^{*(\to j)}_{k-1}=|{\cal C}_{k-1}^{(\to j)}|^{-1}\sum_{l\in {\cal C}_{k-1}^{(\to j)}}\theta^*_{k-1,l}$. We further use the full conditional distribution for the elements in $\bm{\alpha}$, reported in Equation~\ref{SM-eq:full_cond_alpha} (Equation~\ref{eq:full_cond_alpha} in the article).
The full conditional distributions for the remaining parameters are
\begin{align*}
    \sigma^{-2} \mid - &\sim \Ga\left( a_\sigma+\frac{1}{2}\sum_{i=1}^{n}T_i, b_\sigma+\frac{1}{2}\sum_{i=1}^{n}\sum_{t=1}^{T_i}\left\{y_i(x_{it}) - \mathbf{b}(x_{it})^\top\bm{\theta}_i\right\}^2 \right), \\
    \tau^{-2} \mid - &\sim \Ga\left( a_\tau + \frac{1}{2}\sum_{k=1}^KJ_k, b_\tau+\frac{1}{2}\bigg[\sum_{j=1}^{J_1}\bigg(\theta^*_{1j}\bigg)^2 + \sum_{k=2}^K\sum_{j=1}^{J_k}\bigg(\theta^*_{kj}-\phi\theta^{*(\to j)}_{k-1}\bigg)^2\bigg] \right), \\
    \phi \mid - &\sim \Normal\left( \E[\phi\mid-], \var(\phi\mid-) \right),
\end{align*}
where
\begin{equation*}
    \begin{split}
        \var(\phi\mid-) &= \bigg(s_0^{-2}+\tau^{-2}\sum_{k=2}^K\sum_{j=1}^{J_k}(\theta^{*(\to j)}_{k-1})^2\bigg)^{-1}, \\
        \E[\phi\mid-] &= \bigg(s_0^{-2}m_0+\tau^{-2}\sum_{k=2}^K\sum_{j=1}^{J_k}\theta^*_{kj}\theta^{*(\to j)}_{k-1}\bigg) \var(\phi\mid-).
    \end{split}
\end{equation*}
For $j=1,\ldots,J_k$, $k=2,\ldots,K-1$, we have that $\theta^*_{kj} \mid - \sim \Normal\left( \E[\theta^*_{kj}\mid-], \var(\theta^*_{kj}\mid-) \right)$ with

\footnotesize
\begin{equation*}
    \begin{split}
        \var(\theta^*_{kj}\mid-) &= \left[\tau^{-2}+\tau^{-2}\phi^2 \sum_{j'\in \mathcal{C}_{k+1}^{(j\to)}} |\mathcal{C}_{k}^{(\to j')}|^{-2}+\sigma^{-2}\sum_{i:c_{ik}=j}\sum_{t=1}^{T_i} b_k(x_{it})^2\right]^{-1}, \\
        \E[\theta^*_{kj}\mid-] &= \left[ \tau^{-2}\phi \theta^{*(\to j)}_{k-1} + \tau^{-2}\phi\sum_{j'\in \mathcal{C}_{k+1}^{(j\to)}}|\mathcal{C}_{k}^{(\to j')}|^{-1}\epsilon^{(j)}_{k+1,j'} + \sigma^{-2}\sum_{i:c_{ik}=j}\sum_{t=1}^{T_i}b_k(x_{it})r^{(k)}_{it} \right] \var(\theta^*_{kj}\mid-),
    \end{split}
\end{equation*}
\normalsize
where $\epsilon^{(j)}_{k+1,j'}=\theta^*_{k+1,j'}-\phi|\mathcal{C}_{k}^{(\to j')}|^{-1}\sum_{l\in \mathcal{C}_{k}^{(\to j')}\setminus\{j\}}\theta^*_{kl}$, and $r_{it}^{(k)}=y_i(x_{it}) - \sum_{k'\neq k}b_{k'}(x_{it})\theta^*_{k',c_{ik'}}$.
The quantity $\mathcal{C}_{k+1}^{(j \to)}$ denotes the set $\{j'\in\{1,\ldots,J_{k+1}\}: \sum_{i=1}^n\ind{c_{i,k}=j, c_{i,k+1}=j'}>0\}$ containing the $|\mathcal{C}_{k+1}^{(j \to)}|$ indices of clusters to which cluster $j$ is reallocated when moving from index $k$ to $k+1$.

When $k=1$, $\pr(\theta_{kj}^*\mid\bm{\theta}^*_{k-1},\phi,\tau^2)$ is the pdf of a Gaussian distribution with mean 0, instead of $\phi\theta^{*(\to j)}_{0}$, since $\bm{\theta}^*_{0}$ does not exist; hence, for $j=1,\ldots,J_1$, the mean and variance of the full conditional distribution become
\begin{equation*}
    \begin{split}
        \var(\theta^*_{1j}\mid-) &= \left[\tau^{-2}+\tau^{-2}\phi^2 \sum_{j'\in \mathcal{C}_{2}^{(j\to)}} |\mathcal{C}_{1}^{(\to j')}|^{-2}+\sigma^{-2}\sum_{i:c_{i1}=j}\sum_{t=1}^{T_i}b_1(x_{it})^2\right]^{-1}, \\
        \E[\theta^*_{1j}\mid-] &= \left[ \tau^{-2}\phi \sum_{j'\in \mathcal{C}_{2}^{(j\to)}}|\mathcal{C}_{1}^{(\to j')}|^{-1}\epsilon^{(j)}_{2,j'} + \sigma^{-2}\sum_{i:c_{i1}=j}\sum_{t=1}^{T_i}b_1(x_{it})r^{(1)}_{it} \right] \var(\theta^*_{1j}\mid-).
    \end{split}
\end{equation*}
When $k=K$, there are no sums in $j'$ since $\bm{\theta}^*_{K+1}$ does not exist; hence, for $j=1,\ldots,J_K$, the mean and variance of the full conditional distribution become
\begin{equation*}
    \begin{split}
        \var(\theta^*_{Kj}\mid-) &= \left[\tau^{-2}+\sigma^{-2}\sum_{i:c_{iK}=j}\sum_{t=1}^{T_i}b_K(x_{it})^2\right]^{-1}, \\
        \E[\theta^*_{Kj}\mid-] &= \left[ \tau^{-2}\phi\theta^{*(\to j)}_{K-1} + \sigma^{-2}\sum_{i:c_{iK}=j}\sum_{t=1}^{T_i}b_K(x_{it})r^{(K)}_{it} \right] \var(\theta^*_{Kj}\mid-).
    \end{split}
\end{equation*}

\subsection{Posterior estimates}
We are interested in computing posterior point estimates for the sequence of partitions and for the cluster-specific parameters, that is, $(\mu^*_{kj},\sigma^{*2}_{kj})$ in Model~\ref{SM-eq:model_nonfunctional} and $\theta^*_{kj}$ in Model~\ref{SM-eq:model_functional}.
A posterior point estimate for the sequence of partitions $\rho_1, \ldots, \rho_K$ is obtained by independently estimating each partition $\rho_k$ at every index $k$. Each of these is estimated using the SALSO algorithm \citep{dahl_SearchAlgorithmsLoss_2022} with the Binder loss \citep{binder_BayesianClusterAnalysis_1978}.
Posterior estimates of the cluster-specific parameters are computed conditionally on the posterior point estimate of $\rho_1, \ldots, \rho_K$. 
Let $\hat{\mu}^*_{kj}$, $\hat{\sigma}^{*2}_{kj}$, $\hat{\theta}^*_{kj}$ be the posterior point estimates of the cluster-specific parameters $\mu^*_{kj}$, $\sigma^{*2}_{kj}$, $\theta^*_{kj}$, then
\begin{align*}
    \hat{\mu}_{kj}^* &= \frac{1}{|\mathcal{B}|}\sum_{b\in\mathcal{B}} \left\{ \frac{1}{|\hat{S}_{kj}|} \sum_{i\in\hat{S}_{kj}} \mu^{*(b)}_{k,c^{(b)}_{ik}} \right\}, \\
    \hat{\sigma}_{kj}^{*2} &= \frac{1}{|\mathcal{B}|}\sum_{b\in\mathcal{B}} \left\{ \frac{1}{|\hat{S}_{kj}|} \sum_{i\in\hat{S}_{kj}} \sigma^{*2(b)}_{k,c^{(b)}_{ik}} \right\}, \\
    \hat{\theta}_{kj}^* &= \frac{1}{|\mathcal{B}|}\sum_{b\in\mathcal{B}} \left\{ \frac{1}{|\hat{S}_{kj}|} \sum_{i\in\hat{S}_{kj}} \theta^{*(b)}_{k,c^{(b)}_{ik}} \right\},
\end{align*}
for $k=1,\ldots,K$, $j=1,\ldots,\hat{J}_{k}$, with $\hat{J}_{k}$ being the number of clusters in the posterior point estimate of $\rho_k$, $\hat{\rho}_k=\{\hat{S}_{k1},\ldots,\hat{S}_{k\hat{J}_k}\}$; the superscript $(b)$ denotes the value assumed by the quantity at the $b$th MCMC sample, and $\mathcal{B}$ is a set of $|\mathcal{B}|$ MCMC samples.

\section{Simulation study}
See \texttt{simulations\_summary.Rmd} or its respective rendered file, \texttt{simulations\_summary.html}.

\section{Venetian Lagoon tide level data analysis}
See \texttt{venice.Rmd} or its respective rendered file, \texttt{venice.html}.

\subsubsection*{Note on hyperparameter and B-spline selection}
Following the Goldilocks principle, we tried several combinations of prior parameters and B-spline functions, and noticed that, at least in this application, the two key quantities that highly influence the number of clusters are the prior parameters $(a_\sigma,b_\sigma)$ and the number of knots of the B-splines; the remaining prior parameters can be left at their uninformative defaults.
In general, we suggest focusing on $(a_\sigma,b_\sigma)$, considering $(a_\sigma,b_\sigma)=(1,1)$ as a starting point since it delivers good results regardless of the other parameter choices. 
As for the number of knots, the choice highly depends on the functional observations and its selection should be guided by application-specific considerations.
The number of knots, or equivalently the number of basis functions $K$, determines the size of the subintervals $\mathcal{D}_k$ of the functional domain $\mathcal{D}$ and therefore sets the minimum scale at which two (or more) curves can overlap.
In particular, larger values induce a finer partition of the domain, allowing for more local clustering patterns, whereas smaller values result in coarser partitions of $\mathcal{D}$, favoring more global behaviors.
From a practical perspective, this suggests choosing the number of knots, or $K$, to match the level of resolution at which meaningful features are expected in the data.

\newpage
\putbib[Bibliography]
\end{bibunit}

\end{document}